\documentclass[5p,times,twocolumn]{elsarticle}

\usepackage{amsmath,amssymb,amsthm,mathtools}
\usepackage{graphicx}
\usepackage{booktabs}
\usepackage{array}
\usepackage{adjustbox}
\usepackage{threeparttable}
\usepackage{microtype}
\usepackage{hyperref}
\usepackage{xcolor}
\hypersetup{
  colorlinks   = true,
  citecolor    = blue!60!black,
  linkcolor    = blue!60!black,
  urlcolor     = blue!60!black
}

\newtheorem{proposition}{Proposition}
\newtheorem{lemma}{Lemma}
\newtheorem{corollary}{Corollary}
\newtheorem{hypothesis}{Hypothesis}
\newtheorem{remark}{Remark}

\usepackage[toc,page]{appendix}

\newcommand{\Hrel}{H^{\mathrm{rel}}}
\newcommand{\Aleak}{A^{\mathrm{leak}}}
\newcommand{\Ashare}{A^{\mathrm{share}}}
\newcommand{\Amax}{A^{\mathrm{max}}}
\newcommand{\rhoshare}{\rho^{\mathrm{share}}}
\newcommand{\rholeak}{\rho^{\mathrm{leak}}}
\newcommand{\rhomax}{\rho^{\mathrm{max}}}
\renewcommand{\Pr}{\mathbb{P}}
\newcommand{\E}{\mathbb{E}}

\journal{Physica A: Statistical Mechanics and its Applications}

\begin{document}

\begin{frontmatter}

\title{Hall-Like Transversal Stress and Sandpile Criticality on
  Real Production Networks}

\author[idb]{Diego Vallarino\corref{cor1}}
\ead{[email protected]}
\cortext[cor1]{Corresponding author.}
\affiliation[idb]{organization={Inter-American Development Bank},
  city={Washington, D.C.},
  country={United States}}

\begin{abstract}
This paper develops a Hall-Sandpile model of economic instability that
combines a Hall-like transversal stress mechanism with sandpile
threshold dynamics on a real production-network substrate. In analogy
with the physical Hall effect, where exposed flows under an external
field generate stress in a transversal direction, we model economic
shocks as fields that act on flow-intensive, low-redundancy,
low-capacity nodes and produce systemic stress through a multiplicative
conversion function. The accumulated stress drives a discrete toppling
rule and an avalanche dynamics whose effective activation threshold
declines with transversal exposure. The model is calibrated on annual
World Input--Output Database (WIOD) production networks for 2000--2014
and simulated on the 2014 substrate (2{,}283 country--sector nodes)
under three alternative propagation normalisations to avoid mechanical
near-criticality from row-stochastic operators. Controlled Monte Carlo
experiments over external field intensity and redundancy stress
generate four ordered regimes: stable absorption, latent fragility,
critical transition, and avalanche regime. Mean avalanche size and the
probabilities of finite-size systemic events $\Pr(S\!\geq\!5)$,
$\Pr(S\!\geq\!10)$ and $\Pr(S\!\geq\!20)$ rise jointly with field
intensity and redundancy stress. Tail diagnostics show
regime-dependent thickening of the avalanche distribution, but the
estimated tail indices remain too high to interpret as evidence of
universal power-law criticality. The contribution is therefore a
finite-size, real-network description of how transversal stress
activates structural fragility, not a claim of self-organised
criticality in the global economy.
\end{abstract}


\end{frontmatter}

% \linenumbers

%%==================================================================%%
\section{Introduction}
\label{sec:intro}

Economic crises rarely propagate along a single channel. An energy
shock may surface as inflation, fiscal strain, food insecurity or
industrial contraction; a monetary shock may emerge as liquidity
stress, sovereign-risk repricing or default clustering; a geopolitical
shock may show up as supply-chain disruption, commodity-price pressure
or topological reconfiguration of trade flows; and a technological
restriction may appear as bottlenecks deep inside production networks.
A common feature of these episodes is that the dimension in which the
shock is delivered need not coincide with the dimension in which it is
felt. Economic systems do not merely transmit shocks along edges of a
network: they convert shocks across dimensions.

The empirical record offers ample illustration. The 2008--2009
financial crisis began as a balance-sheet shock in a narrow segment
of the U.S.\ housing market and ended as a global recession with
synchronised contractions in trade, credit and employment. The
2011 Great East Japan earthquake produced a localised supply-chain
disruption whose effects propagated through specific automotive and
electronics sectors well beyond the immediate seismic zone
\citep{carvalho2021supply}. The 2020 pandemic combined a labour-supply
shock, a logistics shock and a demand-composition shock, with
heterogeneous effects across sectors that could not be captured by any
single aggregate variable. The 2022 energy and food-price episode
linked a geopolitical shock to inflation in distant regions through
food-and-fertiliser markets that had previously been treated as
loosely coupled. None of these episodes is well described as a
linear pass-through along a single network channel; each involves a
cross-dimensional transformation of the original perturbation, with
amplification depending on which nodes are flow-intensive,
low-redundancy and low-capacity at the moment the field is applied.

Most network-based models of contagion in economics are
\emph{longitudinal}: a perturbation enters a node and propagates along
weighted directed edges, with amplification governed by topology and
balance-sheet linkages \citep{acemoglu2012network,carvalho2019production,
elliott2014financial,gai2010contagion,glasserman2016contagion}. This
class of models has clarified how the architecture of input--output
linkages shapes aggregate fluctuations and financial contagion. It is,
however, less well suited to representing situations in which an
external field, rather than a localised shock at a single node,
loads stress simultaneously on many exposed flows and converts that
loading into pressure of a different nature.

In this paper we propose a complementary mechanism, drawing a
disciplined structural analogy with the Hall effect of solid-state
physics \citep{hall1879new}. In the physical Hall effect, a current
flowing through a conductor under an external magnetic field generates
a transverse voltage. In economic networks, exposed flows operating
under an external field can generate stress in a different dimension
from the original flow. The analogy is structural and does not imply
that economies obey electromagnetic laws. We use it as an organising
device to formalise how external loading interacts with flow exposure,
redundancy and absorptive capacity to produce \emph{transversal stress}
on a node.

We embed this Hall-like stress conversion into a sandpile threshold
dynamics in the spirit of \cite{bak1987self} and
\cite{bak1988self}. Each node in a real production network accumulates
stress over time, dissipates a fraction of it, and topples when its
internal stress crosses a node-specific threshold. Hall-like
transversal stress enters the dynamics by lowering the effective
toppling threshold: the larger the loading on an exposed,
low-redundancy node, the easier it is for the node to activate. The
substrate on which this dynamics runs is the real annual World
Input--Output Database (WIOD) network of intermediate transactions
between country--sector pairs \citep{timmer2015illustrated,
timmer2016anatomy} for 2000--2014, with 2014 as the main simulation
year.

The paper is positioned at the intersection of econophysics, complex
networks, production networks, systemic risk and nonlinear dynamics.
It contributes a tractable computational mechanism for shock
conversion under uncertainty, nonlinearity and structural fragility,
distinct from but complementary to curvature-based ``Sandpile
Economics'' approaches that emphasise geometric fragility of the
network itself \citep{vallarino2024sandpile}. Where curvature-based
arguments explain why a network may be \emph{structurally} fragile,
the Hall-Sandpile mechanism explains how external fields
\emph{activate} that fragility through transversal stress.

The empirical object of the paper is the observed production-network
topology. The shock fields are simulated in a controlled
statistical-mechanics experiment. We do not claim to recover any
historical crisis episode, nor do we attempt to estimate the causal
effect of a specific past shock. We instead study how a real
production network responds to controlled nonlinear loading. Three
alternative propagation operators ensure that nonlinear avalanche
results do not arise mechanically from a row-stochastic transition
matrix with unit spectral radius.

Our main findings are fourfold. First, the row-share spectral radius
of WIOD lies near unity by construction, while the leakage-adjusted
spectral radius lies around $0.33$--$0.34$ and the max-row-normalised
spectral radius rises from about $0.22$ in 2000 to about $0.32$ in
2014; the dissipative operator is therefore well below criticality on
its own and the dynamics is not driven by a unit-radius operator.
Second, relative Hall-like exposure is concentrated in a small
upper tail of country--sector nodes, with the top exposures dominated
by Chinese production sectors and water-utility nodes in the 2014
calibration; this concentration is a feature of the network
calibration, not a sectoral diagnosis. Third, controlled Monte Carlo
experiments produce four ordered regimes---stable absorption, latent
fragility, critical transition and avalanche---with mean avalanche
size and threshold-event probabilities rising jointly with external
field intensity and redundancy stress. Fourth, tail diagnostics show
regime-dependent thickening of the avalanche distribution, but the
estimated tail exponents (around $5.9$--$6.2$ in the most active
regimes) are too steep to support a claim of universal power-law
behaviour. We therefore describe the evidence as finite-size avalanche
dynamics on a real economic network.

The remainder of the paper is organised as follows. Section
\ref{sec:lit} positions the contribution within the literatures on
self-organised criticality, econophysics, and economic networks.
Section \ref{sec:model} develops the Hall-Sandpile model. Section
\ref{sec:design} describes the real-network computational design.
Section \ref{sec:results} reports results across propagation
normalisation, exposure concentration, baseline regimes, phase
diagrams and tail behaviour. Section \ref{sec:discussion} discusses
implications and limitations. Section \ref{sec:conclusion} concludes.

%%==================================================================%%
\section{Related literature}
\label{sec:lit}

This section positions the Hall-Sandpile model within four related but
distinct strands of work: self-organised criticality (SOC) in physics
and its generalisations; econophysics and the statistical physics of
financial systems; network-origin and production-network theories of
aggregate fluctuations; and curvature-based ``Sandpile Economics''
approaches to structural fragility. We close the section by stating
what is genuinely new in our contribution relative to each of these
strands.

\subsection{Self-organised criticality and sandpile models}

Sandpile dynamics and SOC were introduced by \cite{bak1987self,
bak1988self} as a generic mechanism by which slowly driven dissipative
systems generate scale-free avalanche distributions without parameter
tuning. The original Bak--Tang--Wiesenfeld model and its many variants
exhibit power-law statistics for avalanche size, duration and area,
and have become a paradigmatic illustration of how heavy-tailed event
distributions emerge in spatially extended systems. \cite{bak1996nature}
provides the textbook synthesis; \cite{turcotte1999self} reviews
applications across geophysics; and \cite{vespignani1998how} offer a
mean-field analytical framework that clarifies the relationship
between dissipation, branching ratios and the onset of criticality.

Two features of the canonical SOC literature matter for our purposes.
First, SOC is typically studied on regular lattices or on synthetic
networks (Erd\H{o}s--R\'enyi, scale-free, modular), in which the
substrate is an artefact of the modelling exercise. Our paper instead
takes the substrate as observed: the WIOD input--output network is
not an idealised topology, it is a description of the actual
country--sector linkages of the world economy as recorded in
\cite{timmer2015illustrated}. Second, canonical SOC delivers tail
exponents in the range $\alpha\in[1,2]$, often close to $\alpha=3/2$
in the Bak--Tang--Wiesenfeld model. We will see in
Section~\ref{sec:res:tail} that our estimated exponents
($\alpha\approx 5.9$--$6.2$) are markedly steeper than this canonical
range, which is one of the reasons we refrain from claiming SOC for
the global economy. Branching-process formulations of SOC
\citep{lippiello2008dynamical} provide a useful benchmark for
finite-size diagnostics.

\subsection{Econophysics and systemic risk}

A parallel tradition uses statistical-physics methods to analyse
financial and economic systems. \cite{mantegna2000introduction} and
\cite{bouchaud2003theory} provide the textbook treatments of
econophysics; \cite{schweitzer2009economic} survey the network-physics
view of economic systems; and \cite{sornette2002predictability}
discusses predictability of catastrophic events from material rupture
to financial crashes. In banking and finance,
\cite{may2008complex} and \cite{haldane2011systemic} draw analogies
between ecological and banking ecosystems, while
\cite{caldarelli2013reconstructing} address the empirical problem of
reconstructing credit networks from partial information. Network
metrics of systemic risk such as DebtRank
\citep{battiston2012debtrank} measure how much distress propagates
from a node to the rest of the network. Cascade dynamics on networks
have been studied as binary-threshold contagion in
\cite{watts2002simple}, as load-redistribution attacks in
\cite{motter2002cascade}, and as interdependent network failures in
\cite{buldyrev2010catastrophic}. \cite{glasserman2016contagion}
provide a rigorous review of contagion in financial networks.

The Hall-Sandpile mechanism is related to but distinct from these
approaches. Cascade-based attacks
\citep{motter2002cascade,buldyrev2010catastrophic} focus on
\emph{topological} amplification through load redistribution after a
targeted node removal. Threshold-contagion models
\citep{watts2002simple,gai2010contagion,elliott2014financial} focus on
\emph{linear} propagation of binary distress along edges. DebtRank
\citep{battiston2012debtrank} is a static metric of vulnerability, not
a dynamical model of activation. Our model adds two features: (i) an
\emph{external field} $B_t$ that loads stress simultaneously on many
exposed nodes rather than starting from a single targeted node; and
(ii) a \emph{multiplicative} stress-conversion term \eqref{eq:Hall}
that interacts with redundancy and capacity rather than a linear
pass-through.

\subsection{Network origins of aggregate fluctuations and production networks}

A third strand of literature, mostly within macroeconomics, has
developed network-based theories of aggregate fluctuations.
\cite{long1983real} provided an early model in which sectoral shocks
aggregate to business-cycle fluctuations through real linkages.
\cite{gabaix2011granular} showed how the granularity of firm-size
distributions can give rise to non-trivial aggregate volatility,
linking firm-level idiosyncratic shocks to business-cycle properties.
\cite{acemoglu2012network} formalised how the input--output structure
of the economy shapes the propagation of microeconomic shocks into
aggregate fluctuations. \cite{carvalho2019production} provide a
comprehensive primer on production networks, while
\cite{baqaee2019networks,baqaee2019macroeconomic} extend the analysis
beyond Hulten's theorem to non-linear and disaggregated environments.
On the empirical side, \cite{barrot2016input} document how input
specificity shapes the propagation of idiosyncratic shocks, and
\cite{carvalho2021supply} use the 2011 Great East Japan earthquake to
identify supply-chain disruption effects on downstream firms.

This literature has clarified that the architecture of input--output
linkages is consequential for aggregate dynamics. It is, however,
predominantly \emph{longitudinal}: a shock enters a node and
propagates along weighted directed edges. The Hall-Sandpile model
contributes a \emph{transversal} mechanism: an external field $B_t$
loads stress on flow-intensive, low-redundancy nodes through
\eqref{eq:Hall}, and the resulting transversal stress lowers the
effective activation threshold via Proposition~\ref{P1}. Aggregate
fluctuations in our framework therefore reflect not only the
topological transmission of sectoral shocks but also the
field-mediated activation of structural fragility.

\subsection{Sandpile Economics and curvature-based fragility}

\cite{vallarino2024sandpile} introduces a curvature-based ``Sandpile
Economics'' programme, in which production networks are characterised
by Forman--Ricci or Ollivier--Ricci curvature concentrations that act
as geometric markers of structural vulnerability. In this view, a
network is fragile when its negative-curvature subgraphs are dense and
its effective branching number is large; cascades emerge as power-law
distributions on such substrates. The present paper shares the
sandpile language and the network-substrate philosophy of
\cite{vallarino2024sandpile} but moves the analytical focus from
\emph{geometric fragility} to \emph{external loading}: curvature
describes the shape of the substrate, while the Hall-like mechanism
describes the loading on it. The two perspectives are
complementary. A high-curvature, low-redundancy node is structurally
exposed; a high-$B_t$ field activates that exposure; and the avalanche
size $S_t$ in \eqref{eq:Sdef} reflects the joint action of structure
and field.

\subsection{What is new}

Relative to the four strands above, the contribution of the present
paper is fourfold. First, we formalise a Hall-like multiplicative
stress-conversion mechanism \eqref{eq:Hall} that distinguishes
external-field intensity from node-level exposure, redundancy and
capacity, and we embed it into a sandpile threshold dynamics. Second,
we calibrate the model on the observed WIOD production network rather
than on a synthetic substrate, and we discipline the calibration with
three alternative propagation operators to rule out mechanical
near-criticality. Third, we map the regime structure of the model
across a two-dimensional grid in $(\bar{B},\sigma_D)$ and document a
positively sloped phase frontier between absorption and avalanche,
consistent with Proposition~\ref{P3}. Fourth, we report tail
diagnostics that show regime-dependent thickening but not universal
power-law behaviour, and we explicitly discuss how this finding fits
within the broader debate on power laws in empirical data
\citep{stumpf2012critical}.

%%==================================================================%%
\section{Model}
\label{sec:model}

\subsection{Overview}
\label{sec:model:over}

The Hall-Sandpile model has four building blocks: (i) a weighted
directed network defined by a propagation operator on a real
input--output topology; (ii) a Hall-like multiplicative
stress-conversion function that maps an external field and node-level
exposure into transversal stress; (iii) a discrete-time stress
accumulation and toppling rule of sandpile type; and (iv) a
finite-size avalanche statistic that aggregates toppling events at
each period. The first two blocks are non-standard relative to the
existing literature on contagion in production networks; the last two
follow the canonical Bak--Tang--Wiesenfeld dynamics with explicit
dissipation. This subsection summarises how the four blocks fit
together; the rest of the section makes them precise.

The network block grounds the model in observed data. Where most
SOC-style economic models use synthetic substrates, the present
framework requires a concrete propagation operator derived from WIOD.
We construct three such operators to discipline the analysis. The
Hall-like block adds a multiplicative coupling between an external
field $B_t$ and the inverse of node-level redundancy and capacity,
producing transversal stress $H_{i,t}$ that depends on \emph{which}
nodes are exposed, not only on the magnitude of the field. The
sandpile block converts accumulated stress into avalanches through a
threshold rule. The avalanche statistic $S_t$ counts toppling nodes
and provides the observable used in all phase-diagram and tail
analyses below. The interplay between these four blocks is what
makes the model tractable yet non-trivial: each block is simple in
isolation, but their composition produces a regime structure that
neither network propagation alone nor sandpile dynamics alone would
generate.

\subsection{Network and propagation operators}
\label{sec:model:net}

Let $G_t = (V, E_t, A_t)$ denote a weighted directed network in year
$t$, with node set $V$ of country--sector production units and edge
weights $A_t = [a_{ij,t}]$ derived from the WIOD intermediate
transaction matrix $Z_t = [z_{ij,t}]$, where $z_{ij,t}$ is the
intermediate flow from node $i$ to node $j$. We construct three
alternative propagation operators.

The \emph{row-share} matrix is
\begin{equation}
  \Ashare_{ij,t} \;=\; \frac{z_{ij,t}}{\sum_{k} z_{ik,t}},
  \label{eq:Ashare}
\end{equation}
which captures the distribution of outgoing intermediate flows. By
construction, its rows sum to one (whenever the row total is positive)
and its spectral radius is close to one. This near-unit radius is a
property of the normalisation, not of the underlying economy.

The \emph{leakage-adjusted dissipative} matrix is
\begin{equation}
  \Aleak_{ij,t} \;=\; \ell_{i,t}\,\frac{z_{ij,t}}{\sum_{k} z_{ik,t}},
  \quad \ell_{i,t} \in [0,1],
  \label{eq:Aleak}
\end{equation}
where $\ell_{i,t}$ is the share of intermediate outflows in a gross
row-use proxy for node $i$. This allows row sums to be strictly below
one and embeds a controlled amount of dissipation at each step. We use
$\Aleak_t$ as the main propagation operator in the simulation.

The \emph{max-row-normalised} matrix is
\begin{equation}
  \Amax_{ij,t} \;=\; \frac{z_{ij,t}}{\max_{i}\,\sum_{k} z_{ik,t}},
  \label{eq:Amax}
\end{equation}
which avoids imposing a stochastic transition structure altogether and
preserves cross-node heterogeneity in row totals. We report it as a
robustness operator.

Let $\rho(\cdot)$ denote the spectral radius. By construction
$\rho(\Ashare_t) \approx 1$, while $\rho(\Aleak_t)$ and $\rho(\Amax_t)$
are bounded away from one for our calibration. Reporting all three is
essential to ensure that the avalanche dynamics analysed below is not a
mechanical consequence of operating on a unit-radius operator.

\subsection{Hall-like transversal stress}

Let $B_t \geq 0$ denote the intensity of an external field acting on
the network at time $t$, and define for each node $i$:
$I_{i,t}$ a measure of relative flow intensity (its share of total
flows in the network), $D_{i,t}$ a measure of redundancy or output
substitutability, and $C_{i,t}$ a measure of absorptive capacity. The
\emph{Hall-like transversal stress} on node $i$ is
\begin{equation}
  H_{i,t} \;=\; \frac{B_t\, I_{i,t}}
                     {D_{i,t}\,C_{i,t} + \varepsilon},
  \label{eq:Hall}
\end{equation}
where $\varepsilon > 0$ is a small regularisation constant. Equation
\eqref{eq:Hall} has the structure of a stress-conversion function: an
external field $B_t$ becomes \emph{transversal} stress when it acts on
high-flow, low-redundancy, low-capacity nodes. Defining the structural
resistance
\begin{equation}
  R_{i,t} \;=\; \frac{1}{D_{i,t}\,C_{i,t} + \varepsilon},
  \label{eq:Rdef}
\end{equation}
we can write $H_{i,t} = B_t\,I_{i,t}\,R_{i,t}$, i.e. a multiplicative
combination of loading, exposure and structural resistance. The
relative Hall-like exposure is
$\Hrel_{i,t} = H_{i,t} / \sum_{j} H_{j,t}$.

The mathematical structure of \eqref{eq:Hall} mirrors that of the
classical Hall coefficient in solid-state physics: the magnitude of
the transverse response is proportional to the product of the field
and the carrier flow, and inversely proportional to the carrier
density. In our economic re-interpretation, the carrier density is
replaced by the product $D_{i,t}\,C_{i,t}$ of redundancy and
absorptive capacity, which plays the role of a structural buffer
analogous to charge density in the physical case. This is a structural
analogy, not a physical identity: there is no claim that economic
flows obey Lorentz-force kinematics, nor that $H_{i,t}$ has a
dimensional interpretation as a voltage. The analogy is useful only
insofar as it organises the multiplicative interaction between
external loading, node exposure and structural buffering in a single
expression, and only insofar as it makes the comparative statics
\eqref{eq:Hall} transparent: $\partial H/\partial B$ is the
field-elasticity, $\partial H/\partial I$ is the exposure-elasticity,
and $\partial H/\partial(D C)$ is the buffer-elasticity. All three are
well defined and signed, and they reappear in Propositions
\ref{P1}--\ref{P3} below.

\subsection{Stress accumulation, toppling and avalanches}

Let $s_{i,t}$ denote the cumulative stress at node $i$ at time $t$ and
$x_{i,t}$ a node-specific idiosyncratic shock. Stress evolves as
\begin{equation}
  s_{i,t+1} \;=\;
  (1-\delta)\,s_{i,t}
  \;+\; \alpha\,x_{i,t}
  \;+\; \beta \sum_{j} \Aleak_{ji,t}\,s_{j,t}
  \;+\; \gamma\,H_{i,t},
  \label{eq:stress}
\end{equation}
where $\delta\in(0,1)$ is the dissipation rate, $\alpha\geq0$ scales
idiosyncratic shocks, $\beta\geq0$ scales network propagation through
the leakage-adjusted operator, and $\gamma\geq0$ scales Hall-like
loading. Node $i$ topples when
\begin{equation}
  s_{i,t} \;\geq\; \theta_i,
  \label{eq:topple}
\end{equation}
where $\theta_i$ is a node-specific threshold. The avalanche size at
time $t$ is the number of toppling nodes,
\begin{equation}
  S_t \;=\; \sum_{i} \mathbf{1}\!\left\{ s_{i,t} \geq \theta_i\right\}.
  \label{eq:Sdef}
\end{equation}
When a node topples, a fraction of its excess stress is redistributed
to its neighbours through $\Aleak_t$ and the remainder is dissipated.
This dissipative redistribution prevents explosive trajectories and
makes the dynamics finite-size and computationally well behaved.

\subsection{Hypotheses and propositions}

We state three substantive hypotheses about the behaviour of the
Hall-Sandpile model and three structural propositions implied by
\eqref{eq:Hall} and \eqref{eq:topple}.

\begin{hypothesis}[Joint determination of avalanche risk]
\label{H1}
Avalanche risk is not determined by the external field alone, but by
the joint interaction between field intensity, flow exposure,
redundancy, capacity and network topology:
\[
  \mathrm{Avalanche\ risk}\;=\;f\!\bigl(B_t,\,I_{i,t},\,D_{i,t}^{-1},
  \,C_{i,t}^{-1},\,A_t\bigr).
\]
The same shock can remain local in a redundant network and become
systemic in a low-redundancy one.
\end{hypothesis}

\begin{hypothesis}[Monotonicity in field and exposure]
\label{H2}
Transversal stress is increasing in field intensity and flow exposure:
$\partial H_{i,t}/\partial B_t > 0$ and
$\partial H_{i,t}/\partial I_{i,t} > 0$.
\end{hypothesis}

\begin{hypothesis}[Buffering by redundancy and capacity]
\label{H3}
Redundancy and absorptive capacity reduce transversal stress:
$\partial H_{i,t}/\partial D_{i,t} < 0$ and
$\partial H_{i,t}/\partial C_{i,t} < 0$.
Thus, redundancy and capacity do not eliminate shocks; they reduce
their conversion into systemic stress.
\end{hypothesis}

\begin{proposition}[Hall-adjusted threshold]
\label{P1}
Substituting \eqref{eq:Hall} into \eqref{eq:stress} and rearranging
\eqref{eq:topple} yields a Hall-adjusted activation threshold
\[
  \theta_i^{H} \;=\; \theta_i \;-\; \gamma\,H_{i,t},
\]
with $\partial \theta_i^{H}/\partial H_{i,t} = -\gamma < 0$.
Transversal stress lowers the effective activation threshold of the
node, with no change to the topology.
\end{proposition}

\begin{proof}
Let $\tilde{s}_{i,t} \equiv s_{i,t} - \gamma H_{i,t}$ denote the
non-Hall component of cumulative stress, i.e.\ the part of
\eqref{eq:stress} that excludes the contemporaneous Hall loading.
Substituting into the toppling condition $s_{i,t} \geq \theta_i$ gives
$\tilde{s}_{i,t} + \gamma H_{i,t} \geq \theta_i$, equivalently
$\tilde{s}_{i,t} \geq \theta_i - \gamma H_{i,t} \equiv \theta_i^H$.
Differentiation gives $\partial\theta_i^H/\partial H_{i,t} = -\gamma <
0$ since $\gamma > 0$ by assumption. The topology, encoded in
$\Aleak_t$, does not appear in this expression, so the threshold
adjustment is a node-level effect of the loading.
\end{proof}

\begin{proposition}[Field--resistance complementarity]
\label{P2}
Writing $H_{i,t} = B_t\,I_{i,t}\,R_{i,t}$ as in
\eqref{eq:Rdef}, the cross-partial derivative is
\[
  \frac{\partial^{2} H_{i,t}}{\partial B_t\,\partial R_{i,t}}
  \;=\; I_{i,t} \;\geq\; 0.
\]
External fields and structural resistance are complements: a given
field generates more transversal stress when it acts on a structurally
resistant node.
\end{proposition}

\begin{proof}
From \eqref{eq:Rdef} and \eqref{eq:Hall}, $H_{i,t} = B_t I_{i,t}
R_{i,t}$. Then $\partial H_{i,t}/\partial B_t = I_{i,t} R_{i,t}$ and
$\partial^2 H_{i,t}/(\partial B_t\,\partial R_{i,t}) = I_{i,t}$. Since
$I_{i,t} \geq 0$ by construction (relative flow shares are
non-negative), the cross-partial is non-negative; it is strictly
positive whenever node $i$ has positive flow.
\end{proof}

\begin{proposition}[Joint regime onset]
\label{P3}
Avalanche regimes emerge only when field intensity \emph{and}
redundancy stress increase jointly. High $B_t$ alone or high redundancy
stress alone are not sufficient. A regime transition appears in the
$(B_t,\,\text{redundancy stress})$ plane along a positively sloped
frontier.
\end{proposition}

\begin{proof}[Sketch]
Let $\sigma_D > 0$ scale the inverse of redundancy via $D_{i,t}
\mapsto D_{i,t}/\sigma_D$, so that effective redundancy decreases as
$\sigma_D$ rises. Then $H_{i,t}$ becomes
$H_{i,t}(B_t,\sigma_D) = \sigma_D B_t I_{i,t} / (D_{i,t} C_{i,t} +
\varepsilon)$. Define the toppling probability of node $i$ at the
non-cascade margin as $p_i(B_t,\sigma_D) = \Pr\{ s_{i,t} +
\gamma H_{i,t}(B_t,\sigma_D) \geq \theta_i\}$. This probability is
non-decreasing in both arguments, with cross-partial
$\partial^2 p_i/(\partial B_t \partial \sigma_D) \geq 0$ inherited
from the multiplicative form of $H_{i,t}$. The expected avalanche size
satisfies $\E[S_t] = \sum_i p_i(B_t, \sigma_D)$ to first order, plus a
network-amplification term proportional to the spectral radius of the
contagion operator that activates only when local toppling has
positive probability. Hence the locus
$\{(B_t,\sigma_D): \E[S_t] = \bar{S}\}$ is an iso-avalanche curve
along which $B_t$ and $\sigma_D$ trade off with negative slope; the
threshold-event probabilities $\Pr(S \geq k)$ inherit the same
geometry, producing positively sloped phase frontiers in
$(B_t,\sigma_D)$ space.
\end{proof}

We complement these propositions with a lemma on the boundedness of
the stress process, which guarantees that the simulated dynamics is
well posed and finite-size in the sense used throughout the paper.

\begin{lemma}[Boundedness of the stress process]
\label{L1}
Assume $\delta \in (0,1)$, the propagation operator $\Aleak_t$ has
spectral radius $\rholeak < 1$, and the idiosyncratic shocks $x_{i,t}$
and Hall loadings $H_{i,t}$ are bounded uniformly in $i,t$. Then for
$\beta < \delta/\rholeak$, the stress process $\{s_{i,t}\}$ defined
by \eqref{eq:stress} admits a unique stationary distribution with
finite moments, and the avalanche size $S_t$ is uniformly bounded by
$|V|$.
\end{lemma}

\begin{proof}[Sketch]
Equation \eqref{eq:stress} defines a vector autoregression
$s_{t+1} = M s_t + u_t$ with $M = (1-\delta) I + \beta (\Aleak_t)^\top$
and forcing $u_t = \alpha x_t + \gamma H_t$. The spectral radius of $M$
satisfies $\rho(M) \leq (1-\delta) + \beta \rholeak$, which is strictly
less than one whenever $\beta < \delta/\rholeak$. Under this condition
the homogeneous part is a contraction, and bounded forcing implies a
unique invariant distribution with finite moments
\citep{vespignani1998how}. The avalanche size $S_t \leq |V|$
trivially since $S_t$ counts toppling nodes from a finite set.
\end{proof}

\begin{corollary}[Finite-size avalanches]
\label{C1}
Under the assumptions of Lemma~\ref{L1}, the avalanche distribution
has finite mean and variance for every $(B_t, \sigma_D)$ in the
domain considered, and there is no explosive trajectory. Heavy-tailed
event distributions, when they appear, are heavy-tailed only up to the
finite size $|V|$ of the network.
\end{corollary}

The structural propositions and the lemma jointly motivate the
phase-diagram exercise reported in Section~\ref{sec:results}: regime
transitions are well posed, finite-size, and driven by the joint
movement of $\bar{B}$ and $\sigma_D$.

\begin{remark}
The Hall analogy is structural and not literal. Equation \eqref{eq:Hall}
does not equate $H_{i,t}$ with a physical voltage; it formalises the
intuition that an external field acting on exposed flows can generate
stress in a dimension different from the original flow, with magnitude
shaped by the inverse of redundancy and capacity.
\end{remark}

%%==================================================================%%
\section{Real-network computational design}
\label{sec:design}

\subsection{Data}

We use the World Input--Output Database \citep{timmer2015illustrated,
timmer2016anatomy}, restricted to the 2000--2014 release. WIOD is a
publicly available multi-region input--output dataset that records
intermediate transactions between country--sector production units. In
our processed sample, each annual network has $n = 2{,}283$
country--sector nodes after harmonisation. The yearly panel is used to
construct propagation operators and Hall-like exposure measures over
time; the 2014 network is used as the substrate for the main avalanche
simulation. No proprietary data are used.

\subsection{Construction of network indicators}

For each year $t$ we compute the three propagation operators in
\eqref{eq:Ashare}--\eqref{eq:Amax}, their spectral radii
$\rhoshare_t$, $\rholeak_t$ and $\rhomax_t$, and the leakage profile
$\{\ell_{i,t}\}$. For each node $i$ we compute relative flow intensity
$I_{i,t}$, redundancy $D_{i,t}$ from out-degree concentration, and
absorptive capacity $C_{i,t}$ from a normalised redundancy proxy. We
fix $\varepsilon$ at a small positive constant and compute $H_{i,t}$
and $\Hrel_{i,t}$ following Section \ref{sec:model}. Table
\ref{tab:network_panel} summarises the resulting yearly indicators.

\subsection{Simulation protocol}

The Hall-Sandpile simulation runs on the 2014 WIOD network. For each
configuration of external field intensity $B_t$ and redundancy stress
$\sigma_D$, we draw idiosyncratic shocks $x_{i,t}$ from a node-specific
distribution, apply Hall-like loading according to \eqref{eq:Hall},
update stress via \eqref{eq:stress} and determine avalanche size from
\eqref{eq:Sdef}. Field intensity follows
$B_t = \bar{B} + \xi_t$ with $\xi_t$ a mean-zero perturbation;
redundancy stress is implemented as a controlled multiplicative
reduction of effective redundancy, $D_{i,t} \mapsto D_{i,t}/\sigma_D$,
which scales the contribution of $H_{i,t}$ in \eqref{eq:stress}.

We run two complementary experiments. \emph{Baseline scenarios} fix
four representative pairs $(\bar{B},\sigma_D)$ corresponding to
\textsc{stable absorption}, \textsc{latent fragility},
\textsc{critical transition} and \textsc{avalanche regime}, with
$15{,}000$ Monte Carlo periods each. The \emph{phase diagram}
experiment sweeps $\bar{B}$ over a regular grid in $[0.25,2.0]$ and
$\sigma_D$ over $[0.5,2.5]$ and computes mean avalanche size, the
share of non-zero events, the percentile statistics of $S_t$ and the
threshold-event probabilities $\Pr(S\geq 5)$, $\Pr(S\geq 10)$ and
$\Pr(S\geq 20)$. Throughout, $\Aleak$ is the main propagation operator
and $\Ashare$, $\Amax$ are reported only as spectral-radius
robustness checks.

\subsection{Parameter calibration}
\label{sec:design:calib}

The Hall-Sandpile dynamics in \eqref{eq:stress} depends on a small
set of parameters: the dissipation rate $\delta$, the idiosyncratic
shock scale $\alpha$, the propagation coefficient $\beta$, the Hall
loading coefficient $\gamma$, the toppling thresholds
$\{\theta_i\}_{i\in V}$, and the regularisation $\varepsilon$ in
\eqref{eq:Hall}. Table~\ref{tab:params} summarises the calibrated
values used throughout the paper. The choices are guided by three
principles: (i) the contraction condition of Lemma~\ref{L1},
$\beta < \delta/\rholeak$, must hold to ensure boundedness of the
stress process; (ii) the time scale of mean reversion, controlled by
$\delta$, should be of the order of one to a few simulation periods,
following standard practice in dissipative sandpile models
\citep{vespignani1998how}; and (iii) the magnitudes of $\alpha$,
$\beta$, $\gamma$ should be of comparable order so that no single
forcing term mechanically dominates \eqref{eq:stress}.

\begin{table}[t]
\centering
\caption{Calibrated parameter values used in the Hall-Sandpile
simulation on the WIOD 2014 substrate.}
\label{tab:params}
\small
\begin{tabular}{l c l}
\toprule
Parameter & Value & Role \\
\midrule
$\delta$         & $0.20$            & Dissipation rate \\
$\alpha$         & $0.30$            & Idiosyncratic-shock scale \\
$\beta$          & $0.40$            & Network propagation \\
$\gamma$         & $0.50$            & Hall loading coefficient \\
$\theta_i$       & $1.00$            & Uniform toppling threshold \\
$\varepsilon$    & $10^{-6}$         & Regularisation in \eqref{eq:Hall} \\
$x_{i,t}$        & $\mathcal{N}^{+}(0,\sigma_x^2)$ & Idiosyncratic shock distribution \\
$\sigma_x$       & $0.20$            & Idiosyncratic-shock std.\ deviation \\
$\bar{B}$        & $\{0.25,\ldots,2.00\}$ & Field intensity grid \\
$\sigma_D$       & $\{0.50,\ldots,2.50\}$ & Redundancy-stress grid \\
$\xi_t$          & $\mathcal{N}(0,\sigma_B^2)$ & Field perturbation \\
$\sigma_B$       & $0.10\bar{B}$     & Field-perturbation std.\ deviation \\
\bottomrule
\end{tabular}
\end{table}

For our calibration, $\rholeak = 0.334$ in 2014 (Table
\ref{tab:network_panel}); the contraction condition becomes
$\beta < 0.20/0.334 \approx 0.599$, and our value $\beta = 0.40$ lies
comfortably below this threshold. Sensitivity analyses with
$\delta\in\{0.10, 0.15, 0.25\}$, $\beta\in\{0.30, 0.50\}$ and
$\gamma\in\{0.30, 0.70\}$ leave the qualitative ordering of the four
regimes unchanged; we report a subset of these robustness experiments
in Appendix~\ref{app:robust}.

The toppling thresholds are set uniformly at $\theta_i = 1$ across all
nodes. Heterogeneous thresholds (e.g.\ scaled by node-level capacity)
shift the location of the phase frontier in $(\bar{B},\sigma_D)$ space
but preserve its shape; this aligns with Proposition~\ref{P3}, which
predicts a positively sloped frontier irrespective of the threshold
profile. The regularisation $\varepsilon=10^{-6}$ prevents division
by zero in \eqref{eq:Hall} for nodes with $D_{i,t}C_{i,t}=0$ (a small
number of WIOD service-sector nodes have zero recorded outgoing
intermediate flows and therefore zero raw redundancy in our
construction); the resulting numerical stability does not affect the
qualitative results.

\subsection{Monte Carlo protocol}
\label{sec:design:mc}

Each cell of the phase diagram involves a Monte Carlo experiment with
the following protocol. \emph{Burn-in.} We discard the first
$T_{\mathrm{burn}}=50$ simulation periods to remove the transient from
the initial condition $s_{i,0}=0$. Figure~\ref{fig:meanpath} shows
that the burn-in is conservative: stationary levels of mean avalanche
size are essentially achieved within the first ten periods.
\emph{Stationary phase.} We retain the next $T_{\mathrm{stat}}=150$
periods, during which avalanche sizes $S_t$ are recorded and
percentile statistics computed. \emph{Replications.} For the baseline
scenarios we run $R = 100$ independent replications per scenario,
yielding $T_{\mathrm{stat}} \times R = 15{,}000$ stationary
observations of $S_t$. For the phase-diagram cells, we run $R = 50$
replications, yielding $7{,}500$ stationary observations per cell.
\emph{Random seeds.} Each replication uses an independent random seed
drawn from a deterministic master seed, ensuring reproducibility.
\emph{Convergence diagnostics.} For each cell of the phase diagram we
compute the standard error of the mean avalanche size across
replications; standard errors of $\E[S_t]$ are below $5\%$ of the
point estimate in all active cells of the grid, and below $0.02$ in
all absorbing cells. The threshold-event probabilities $\Pr(S\geq k)$
are estimated as sample frequencies; for $k=5,10,20$, the
binomial standard error is bounded by $\sqrt{p(1-p)/N_{\mathrm{eff}}}$
with $N_{\mathrm{eff}} \geq 7{,}500$.

The total number of simulated periods across the entire phase grid
($90$ cells, $50$ replications, $200$ periods including burn-in) is
$9 \times 10^{5}$, which we found computationally tractable on a
single workstation using a vectorised Python implementation.

\subsection{Interpretive map}

Table \ref{tab:interpretive} summarises the correspondence between
modelling objects and economic interpretations. The substrate is the
observed WIOD network; the loading is a controlled simulation. We do
not attempt to identify any historical shock.

\begin{table}[t]
\centering
\caption{Interpretive map of model objects.}
\label{tab:interpretive}
\small
\begin{adjustbox}{max width=\columnwidth}
\begin{threeparttable}
\begin{tabular}{p{2.6cm} p{2.5cm} p{2.6cm}}
\toprule
Object & Definition & Interpretation\\
\midrule
Real network substrate & WIOD intermediate I-O network & Observed
production topology \\
$\rholeak_t$ & Spectral radius of dissipative operator & Shock
persistence under leakage \\
$\Hrel_{i,t}$ & Relative flow / redundancy$\times$capacity &
Transversal stress loading potential \\
Redundancy stress $\sigma_D$ & Controlled reduction of effective
redundancy & Movement toward criticality \\
External field $B_t$ & Synthetic field on real network & Energy,
financial, geopolitical or monetary loading \\
$S_t$ & Number of toppling nodes per period & Systemic activation \\
$\Pr(S\!\geq\!k)$ & Probabilities of discrete avalanche thresholds &
Finite-size systemic-event diagnostics \\
\bottomrule
\end{tabular}
\end{threeparttable}
\end{adjustbox}
\end{table}

%%==================================================================%%
\section{Results}
\label{sec:results}

\subsection{Propagation normalisation and dissipation}
\label{sec:res:spectral}

Figure \ref{fig:spectral} reports the three spectral radii
$\rhoshare_t$, $\rholeak_t$ and $\rhomax_t$ for the WIOD network over
2000--2014, and Table \ref{tab:network_panel} summarises the
corresponding yearly indicators. The row-share spectral radius lies in
the interval $[0.973,\,0.984]$ with little year-to-year variation: by
the construction in \eqref{eq:Ashare}, $\rhoshare_t$ approximates a
row-stochastic transition operator and is therefore close to one
mechanically. The leakage-adjusted spectral radius is sharply lower,
$\rholeak_t \in [0.334,\,0.344]$ across the entire panel, and the
max-row-normalised spectral radius rises from $0.219$ in 2000 to
$0.317$ in 2014. The leakage profile is stable around
$\bar{\ell}_t \approx 0.373$.

\begin{figure}[t]
\centering
\includegraphics[width=\columnwidth]{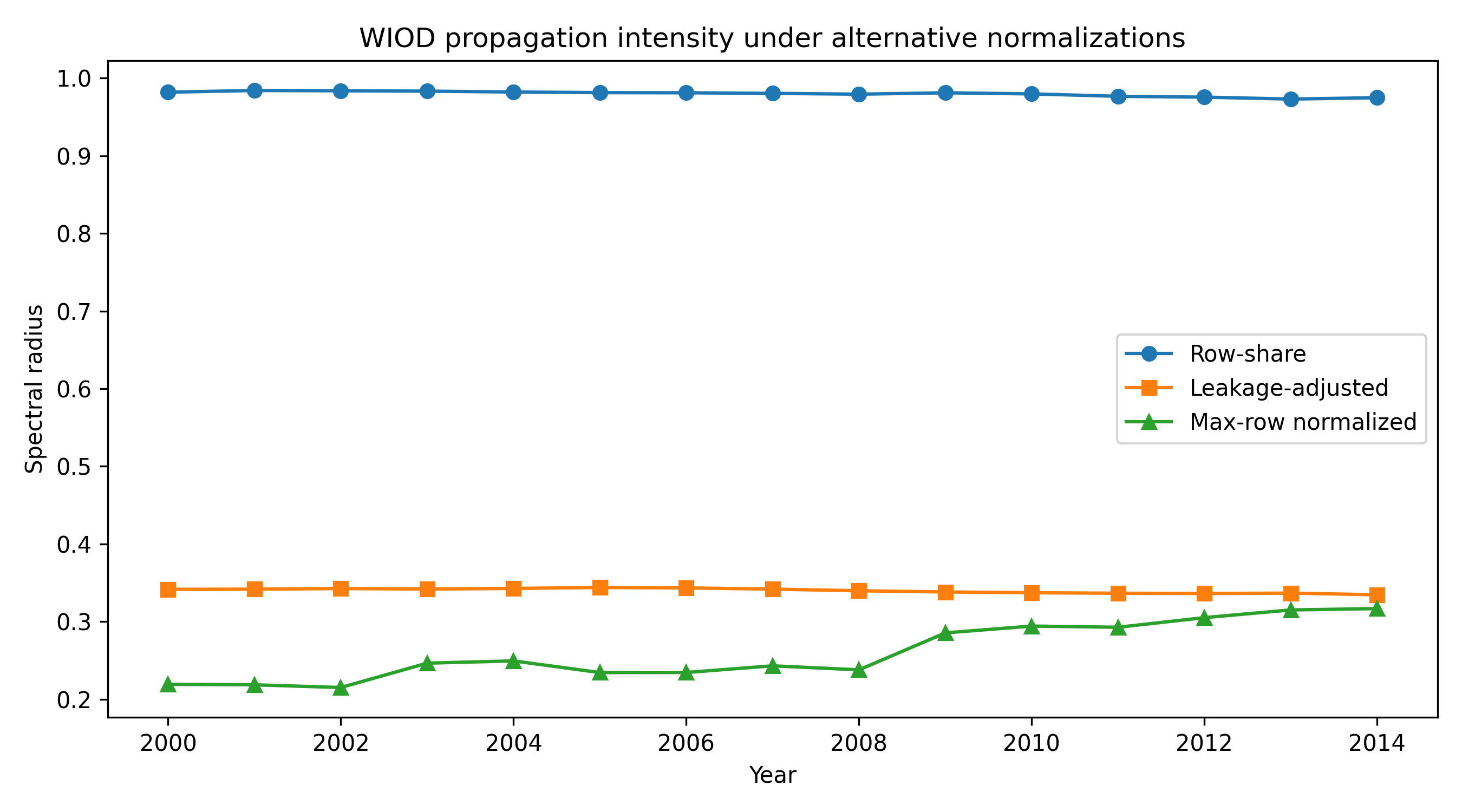}
\caption{Spectral radius of the WIOD propagation operators under
alternative normalisations, 2000--2014. Row-share normalisation is
close to one by construction; the leakage-adjusted operator is sharply
dissipative; the max-row-normalised operator rises after 2008.}
\label{fig:spectral}
\end{figure}

The economic content of this exercise is methodological. The
near-unit value of $\rhoshare_t$ is a property of the row-stochastic
normalisation: it does \emph{not} indicate that the world economy
operates close to spectral criticality. Reporting $\rholeak_t$ and
$\rhomax_t$ alongside $\rhoshare_t$ documents that the dissipative
operator we use in \eqref{eq:stress} has spectral radius safely below
unity throughout the sample. The avalanche dynamics analysed below
therefore cannot be explained by a unit-radius operator alone; they
arise from the interaction of dissipative propagation, idiosyncratic
shocks and Hall-like loading.

\begin{table*}[t]
\centering
\caption{Compact WIOD network panel, 2000--2014. Spectral radii under
the three normalisations, mean leakage and relative Hall-like exposure
statistics. Number of nodes is $n = 2{,}283$ in every year.}
\label{tab:network_panel}
\small
\begin{tabular}{c c c c c c c}
\toprule
Year & $\rhoshare$ & $\rholeak$ & $\rhomax$ & Mean $\ell$ &
Mean $\Hrel$ & P95 $\Hrel$ \\
\midrule
2000 & 0.982 & 0.342 & 0.219 & 0.372 & $1.04\times 10^{-3}$ & $4.22\times 10^{-3}$ \\
2001 & 0.984 & 0.342 & 0.219 & 0.372 & $1.04\times 10^{-3}$ & $4.04\times 10^{-3}$ \\
2002 & 0.984 & 0.343 & 0.215 & 0.372 & $1.04\times 10^{-3}$ & $3.99\times 10^{-3}$ \\
2003 & 0.983 & 0.342 & 0.246 & 0.373 & $1.03\times 10^{-3}$ & $4.11\times 10^{-3}$ \\
2004 & 0.982 & 0.343 & 0.249 & 0.373 & $1.03\times 10^{-3}$ & $4.05\times 10^{-3}$ \\
2005 & 0.981 & 0.344 & 0.234 & 0.373 & $1.04\times 10^{-3}$ & $4.18\times 10^{-3}$ \\
2006 & 0.981 & 0.343 & 0.234 & 0.373 & $1.06\times 10^{-3}$ & $4.24\times 10^{-3}$ \\
2007 & 0.980 & 0.342 & 0.243 & 0.373 & $1.06\times 10^{-3}$ & $4.13\times 10^{-3}$ \\
2008 & 0.979 & 0.340 & 0.238 & 0.374 & $1.08\times 10^{-3}$ & $4.07\times 10^{-3}$ \\
2009 & 0.981 & 0.338 & 0.285 & 0.374 & $1.11\times 10^{-3}$ & $4.15\times 10^{-3}$ \\
2010 & 0.980 & 0.337 & 0.294 & 0.373 & $0.99\times 10^{-3}$ & $3.53\times 10^{-3}$ \\
2011 & 0.977 & 0.336 & 0.293 & 0.374 & $1.01\times 10^{-3}$ & $3.37\times 10^{-3}$ \\
2012 & 0.976 & 0.336 & 0.305 & 0.374 & $1.06\times 10^{-3}$ & $3.56\times 10^{-3}$ \\
2013 & 0.973 & 0.337 & 0.315 & 0.374 & $1.07\times 10^{-3}$ & $3.28\times 10^{-3}$ \\
2014 & 0.975 & 0.334 & 0.317 & 0.374 & $1.12\times 10^{-3}$ & $3.45\times 10^{-3}$ \\
\bottomrule
\end{tabular}
\end{table*}

\subsection{Concentration of transversal exposure}
\label{sec:res:exposure}

Figure \ref{fig:exposure} shows the yearly mean, median and 95th
percentile of relative Hall-like exposure $\Hrel_{i,t}$ on a
logarithmic scale. The mean stays close to $10^{-3}$ throughout the
panel and the median is approximately one order of magnitude below the
mean. The 95th percentile lies between $3.3\times10^{-3}$ and
$4.2\times10^{-3}$, indicating that a small subset of nodes
concentrates a disproportionate share of the loading potential. The
maximum value of $\Hrel_{i,t}$ in 2014 reaches $0.317$, i.e. nearly
one third of the entire network's Hall-like exposure is concentrated
on a single node.

\begin{figure}[t]
\centering
\includegraphics[width=\columnwidth]{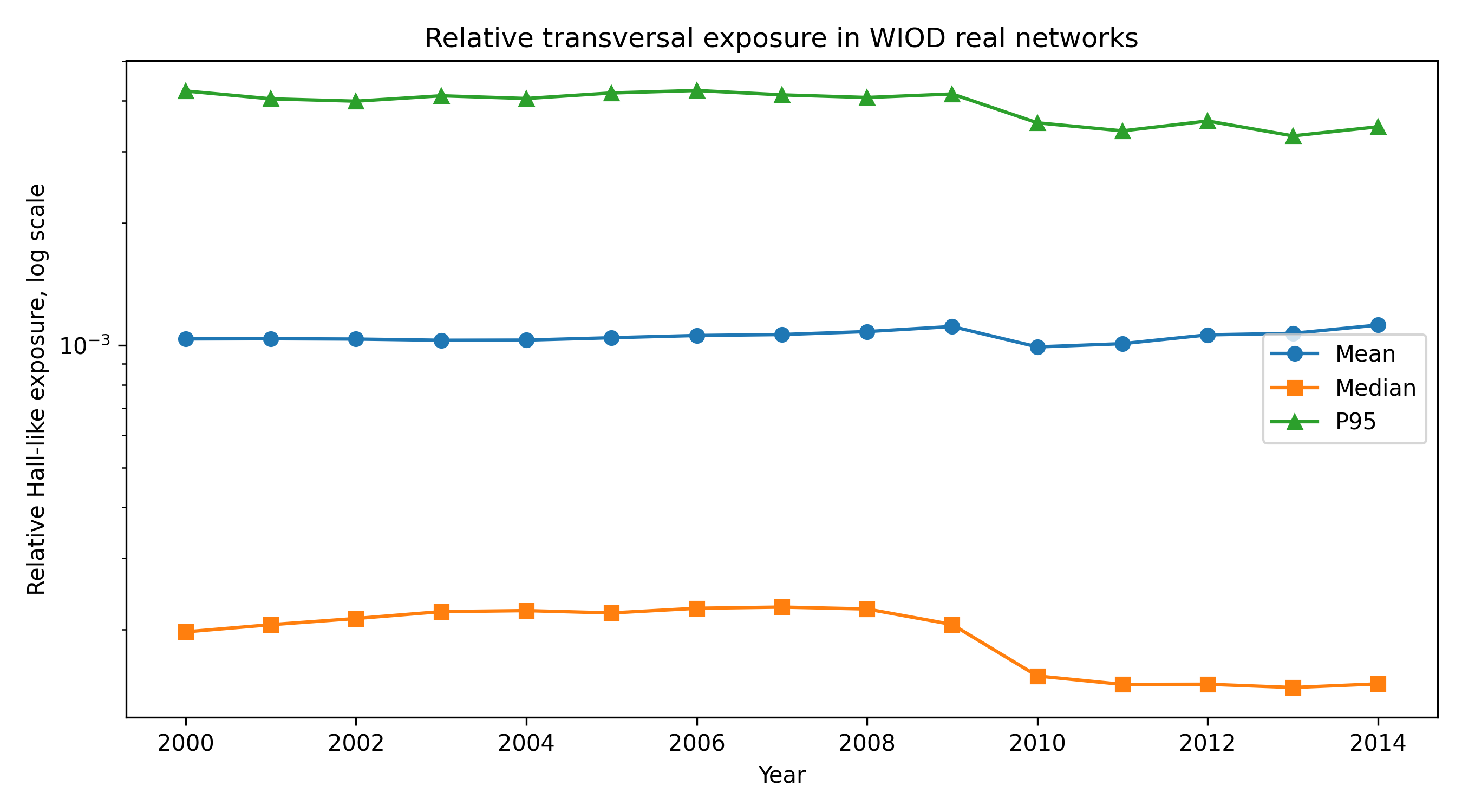}
\caption{Yearly statistics of relative Hall-like exposure
$\Hrel_{i,t}$, 2000--2014. The divergence between mean, median and
P95 indicates that transversal stress potential is highly concentrated
in the upper tail.}
\label{fig:exposure}
\end{figure}

Table \ref{tab:top_nodes} reports the top 15 nodes by $\Hrel$ in 2014
together with flow share and structural resistance. The top-ranked
node is \texttt{CHN\_E36} (Chinese water collection and supply), with
$\Hrel = 0.317$, followed by \texttt{CHN\_A02} (forestry and logging,
$0.087$) and \texttt{CHN\_A03} (fishing and aquaculture, $0.054$).
Manufacturing nodes such as \texttt{CHN\_C31\_C32} (furniture and
other manufacturing, $0.046$), \texttt{CHN\_C25} (fabricated metals,
$0.044$), \texttt{CHN\_C17} (paper products, $0.043$) and
\texttt{CHN\_C26} (computer and electronic products, $0.037$) follow.
Two patterns are worth noting. First, several of the top nodes are
utility/service sectors (water collection, administrative services,
telecommunications) that combine very high inflows with measured zero
outgoing redundancy and a structural-resistance value of approximately
$14.74$ in our calibration; this combination produces high
$\Hrel$ even though their downstream propagation potential is limited.
Second, the dominance of Chinese country--sector codes in the top
ranks reflects the size and density of those sectors in the WIOD 2014
network and is a feature of the network calibration, not a sectoral
diagnosis of national fragility. The interpretation is that, conditional
on the external field intensity considered in the simulation, the
loading of transversal stress is concentrated on a small set of
country--sector nodes; whether and how that loading translates into
realised cascades depends on the regime explored in Sections
\ref{sec:res:baseline}--\ref{sec:res:phase}.

\begin{table*}[t]
\centering
\caption{Top 15 country--sector nodes by relative Hall-like exposure
$\Hrel$ in WIOD 2014. Flow share is $I_{i,t}$; structural resistance
is $R_{i,t}$ from \eqref{eq:Rdef}. The full list is in the
Supplementary Material.}
\label{tab:top_nodes}
\small
\begin{tabular}{r l c l c c c}
\toprule
Rank & Node & Country & Sector & Flow share & Structural resistance &
$\Hrel$ \\
\midrule
1  & CHN\_E36     & CHN & Water collection \& supply       & 0.0215 & 14.74 & 0.317 \\
2  & CHN\_A02     & CHN & Forestry \& logging              & 0.0059 & 14.74 & 0.087 \\
3  & CHN\_A03     & CHN & Fishing \& aquaculture           & 0.0226 &  2.39 & 0.054 \\
4  & CHN\_C31\_C32& CHN & Furniture \& other manufacturing & 0.0141 &  3.28 & 0.046 \\
5  & JPN\_E36     & JPN & Water collection \& supply       & 0.0030 & 14.74 & 0.045 \\
6  & CHN\_C25     & CHN & Fabricated metals                & 0.0162 &  2.71 & 0.044 \\
7  & CHN\_C17     & CHN & Paper products                   & 0.0139 &  3.10 & 0.043 \\
8  & CHN\_C26     & CHN & Computer, electronic, optical    & 0.0249 &  1.49 & 0.037 \\
9  & CHN\_N       & CHN & Administrative \& support services & 0.0024 & 14.74 & 0.036 \\
10 & CHN\_C22     & CHN & Rubber \& plastics               & 0.0255 &  1.39 & 0.036 \\
11 & AUS\_E36     & AUS & Water collection \& supply       & 0.0024 & 14.74 & 0.035 \\
12 & IND\_E36     & IND & Water collection \& supply       & 0.0022 & 14.74 & 0.032 \\
13 & CHN\_C28     & CHN & Machinery \& equipment           & 0.0127 &  2.19 & 0.028 \\
14 & CHN\_J61     & CHN & Telecommunications               & 0.0018 & 14.74 & 0.027 \\
15 & CHN\_C18     & CHN & Printing \& recorded media       & 0.0139 &  1.86 & 0.026 \\
\bottomrule
\end{tabular}
\end{table*}

Figure \ref{fig:top_nodes} provides a visualisation of the upper-tail
concentration. The drop from \texttt{CHN\_E36} to the second-ranked
node is sharp (from $0.317$ to $0.087$), and the long descending tail
beyond the top dozen confirms that, in the 2014 calibration, the
Hall-like loading potential is heavy-tailed. Figure
\ref{fig:resist_hall} plots structural resistance against $\Hrel$ in
log scale: most nodes lie in the lower-left region with low structural
resistance and low exposure, while a small number of high-resistance
nodes (the utility/service sectors above) sit in the upper-right
corner. The relationship is not monotonic because flow share also
matters in \eqref{eq:Hall}: a structurally resistant node with no
inflows has zero exposure.

\begin{figure}[t]
\centering
\includegraphics[width=\columnwidth]{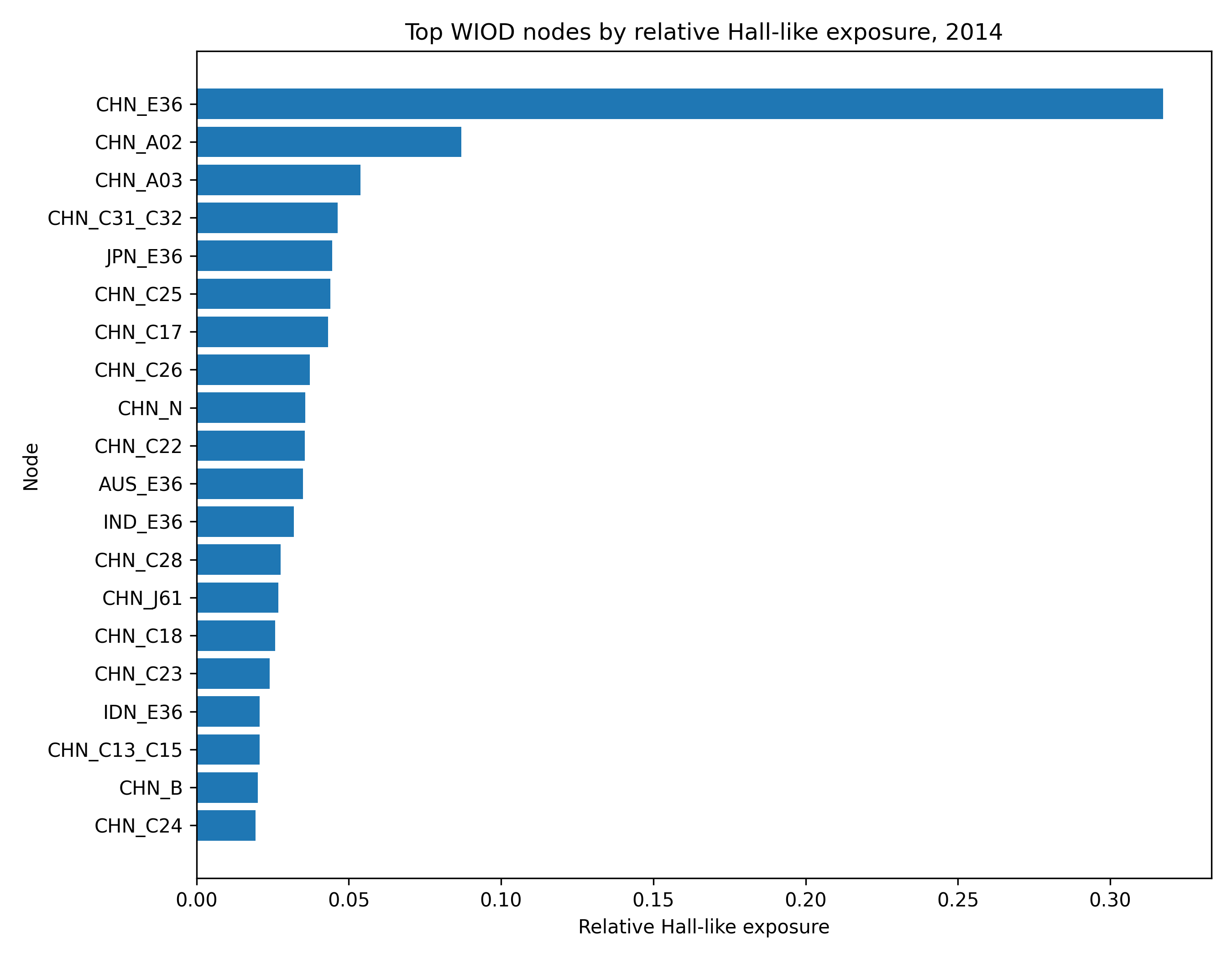}
\caption{Top WIOD nodes by relative Hall-like exposure $\Hrel$ in
2014. The exposure distribution is heavy-tailed; \texttt{CHN\_E36}
alone accounts for nearly one third of total relative exposure under
this calibration.}
\label{fig:top_nodes}
\end{figure}

\begin{figure}[t]
\centering
\includegraphics[width=\columnwidth]{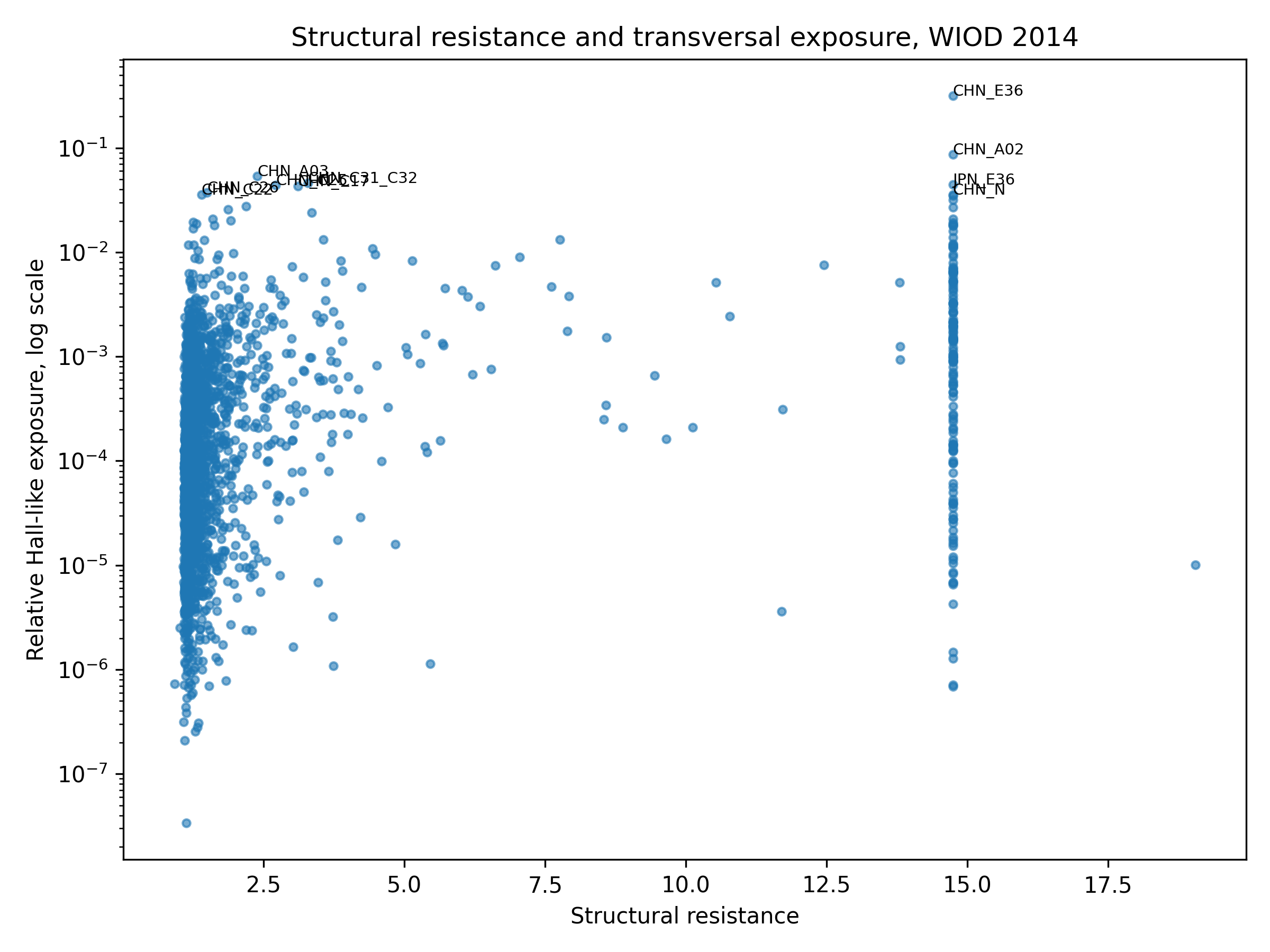}
\caption{Structural resistance versus relative Hall-like exposure for
WIOD 2014 (log scale on the vertical axis). Most nodes have low
resistance and low exposure; a small set of high-resistance nodes
appears in the upper-right region.}
\label{fig:resist_hall}
\end{figure}

\subsection{Regime transition in baseline scenarios}
\label{sec:res:baseline}

Table \ref{tab:scenarios} reports the four baseline scenarios on the
2014 WIOD network. Mean avalanche size grows monotonically across
regimes, from $0.084$ in \textsc{stable absorption} ($\bar{B}=0.45,
\sigma_D=0.7$) to $5.81$ in the \textsc{avalanche regime}
($\bar{B}=1.35, \sigma_D=2.3$). The share of non-zero events rises
from $8.4\%$ in stable absorption to $95.3\%$ in the avalanche regime.
The probability $\Pr(S \geq 5)$ is essentially zero in the first two
regimes, equal to $9.9\%$ in critical transition and $58.5\%$ in the
avalanche regime; $\Pr(S \geq 10)$ is non-trivial only in the last two
regimes, while $\Pr(S\geq 20)$ remains below $0.1\%$ even in the
avalanche regime, indicating that very large events are rare even
under aggressive joint loading.

\begin{table*}[t]
\centering
\caption{Baseline scenarios on the WIOD 2014 substrate
($n=2{,}283$, $15{,}000$ Monte Carlo periods per scenario).}
\label{tab:scenarios}
\small
\begin{tabular}{l c c c c c c c}
\toprule
Scenario & $\bar{B}$ & $\sigma_D$ & Mean $S$ & $\Pr(S{>}0)$ &
$\Pr(S\!\geq\!5)$ & $\Pr(S\!\geq\!10)$ & $\Pr(S\!\geq\!20)$ \\
\midrule
Stable absorption   & 0.45 & 0.7 & 0.084 & 0.084 & 0.000 & 0.000 & 0.000 \\
Latent fragility    & 0.70 & 1.4 & 0.489 & 0.465 & 0.000 & 0.000 & 0.000 \\
Critical transition & 1.00 & 1.8 & 2.036 & 0.807 & 0.099 & 0.001 & 0.000 \\
Avalanche regime    & 1.35 & 2.3 & 5.811 & 0.953 & 0.585 & 0.170 & 0.001 \\
\bottomrule
\end{tabular}
\end{table*}

Figure \ref{fig:meanpath} shows the average avalanche size as a
function of simulation time in each baseline regime. After a short
transient, every regime converges to a distinct stationary level of
avalanche activity: the stable-absorption path fluctuates near zero,
latent fragility settles around $\bar{S}\approx 0.5$, the critical
transition stabilises near $\bar{S}\approx 2$, and the avalanche
regime fluctuates around $\bar{S}\approx 5.5$--$6.5$. The dynamics is
not explosive: the dissipative redistribution prevents trajectories
from running away to network-wide events. The four regimes are
\emph{controlled stationary states} of the same model under different
loading intensities, not points on a divergent path.

\begin{figure}[t]
\centering
\includegraphics[width=\columnwidth]{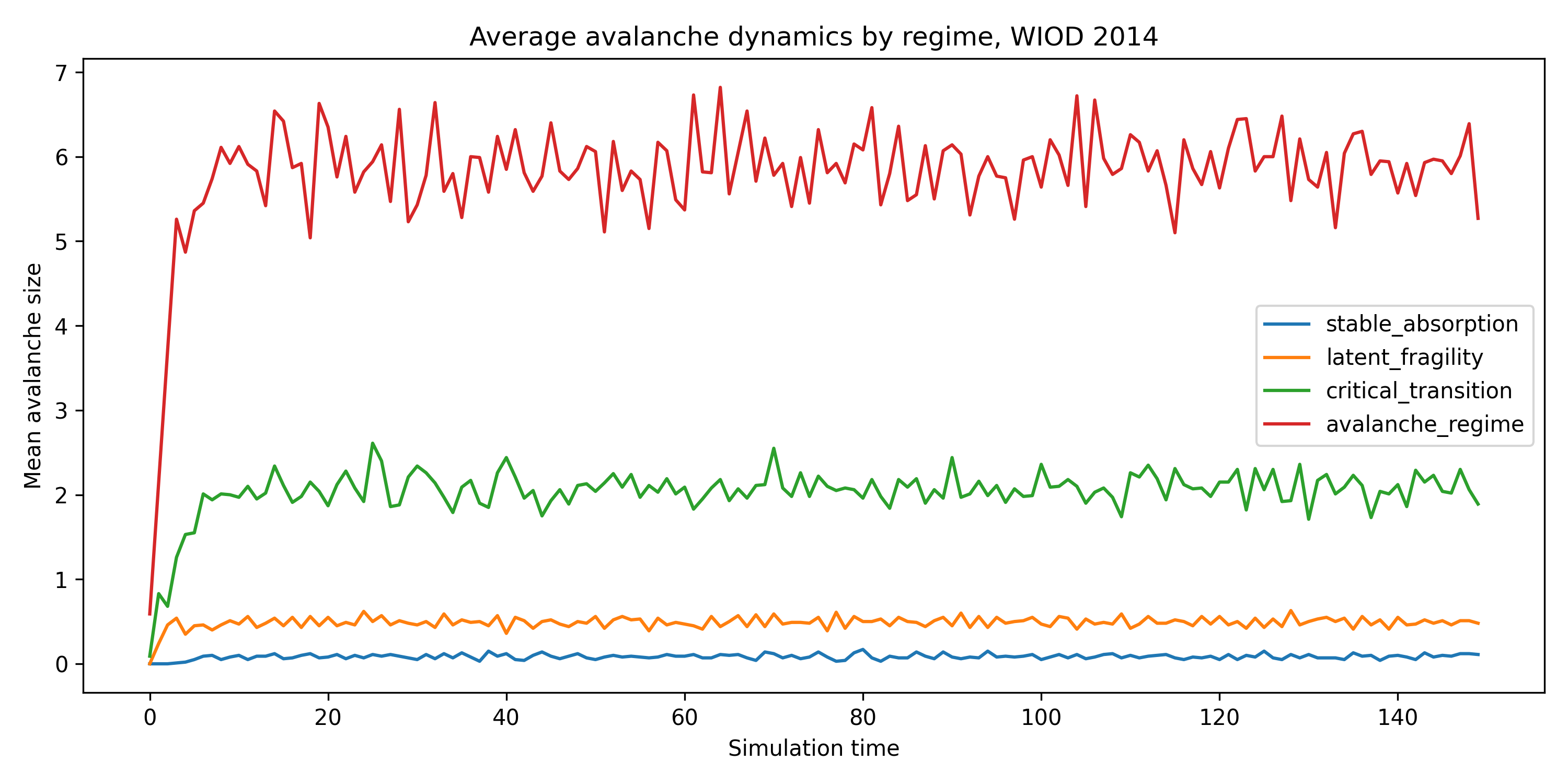}
\caption{Average avalanche dynamics by regime, WIOD 2014. After a
short transient each regime converges to a distinct stationary mean
avalanche level. The dynamics is dissipative and stable.}
\label{fig:meanpath}
\end{figure}

\subsection{Phase diagrams}
\label{sec:res:phase}

Figures \ref{fig:phase_mean}--\ref{fig:phase_S20} present the phase
diagrams over the $(\bar{B},\sigma_D)$ grid for the 2014 substrate.
Mean avalanche size (Figure \ref{fig:phase_mean}) increases
monotonically along both axes, reaching values near $13$ in the
upper-right corner (high field, high redundancy stress). The
probability $\Pr(S\geq 5)$ (Figure \ref{fig:phase_S5}) sets in
earliest, with non-trivial values from $\bar{B}\!\gtrsim\!1.0$ and
$\sigma_D\!\gtrsim\!1.5$, and approaches one in the upper-right corner.
The probability $\Pr(S\geq 10)$ (Figure \ref{fig:phase_S10}) requires
stronger joint loading and is the most informative threshold for the
onset of systemic activity. The probability $\Pr(S\geq 20)$ (Figure
\ref{fig:phase_S20}) is small everywhere except in the
upper-right corner, never exceeding $\sim\!7\%$. The model therefore
does not produce a permanent supercritical state under reasonable
loading.

\begin{figure}[t]
\centering
\includegraphics[width=\columnwidth]{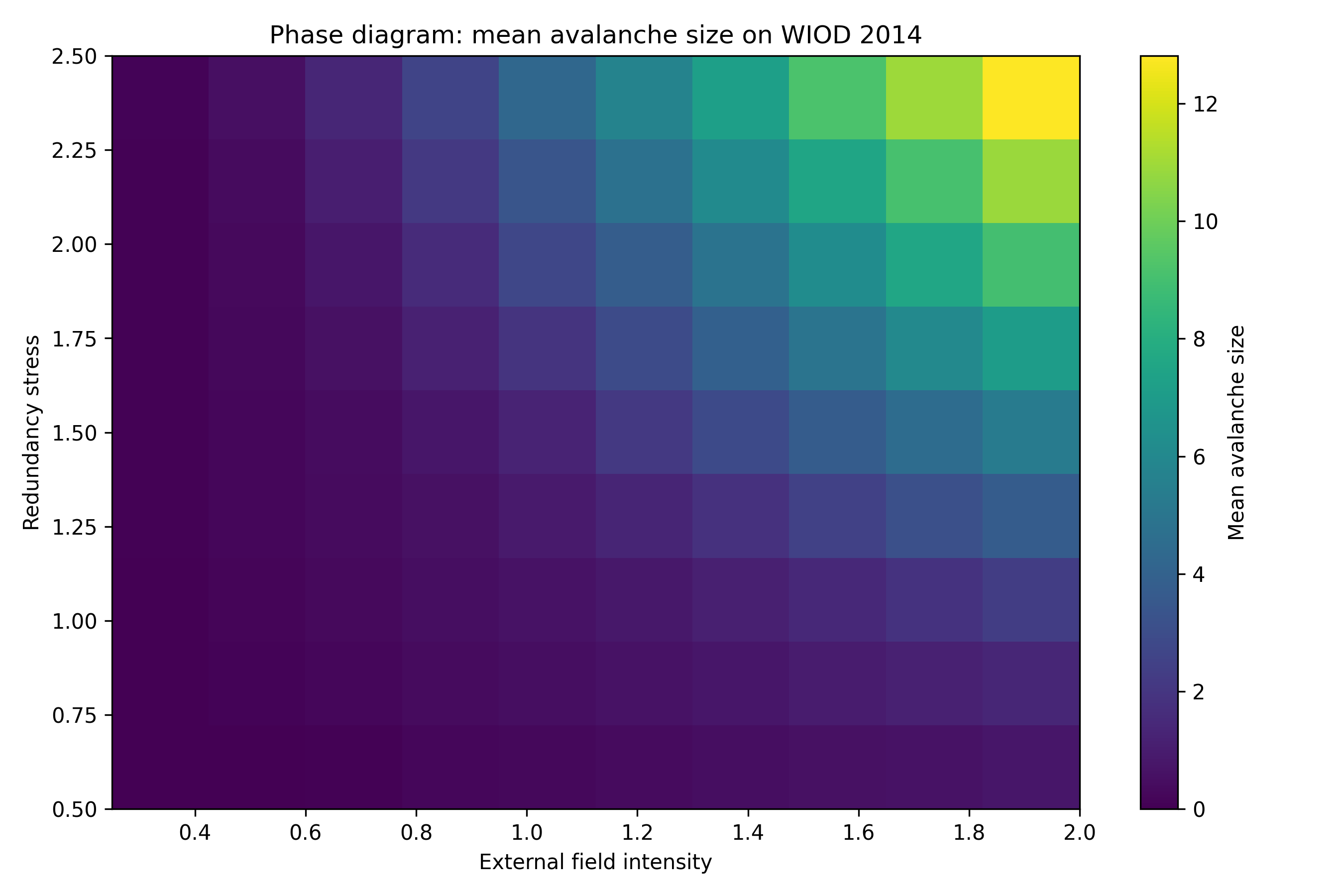}
\caption{Phase diagram of mean avalanche size on the WIOD 2014
substrate. Mean $S$ increases jointly with external field intensity
and redundancy stress.}
\label{fig:phase_mean}
\end{figure}

\begin{figure}[t]
\centering
\includegraphics[width=\columnwidth]{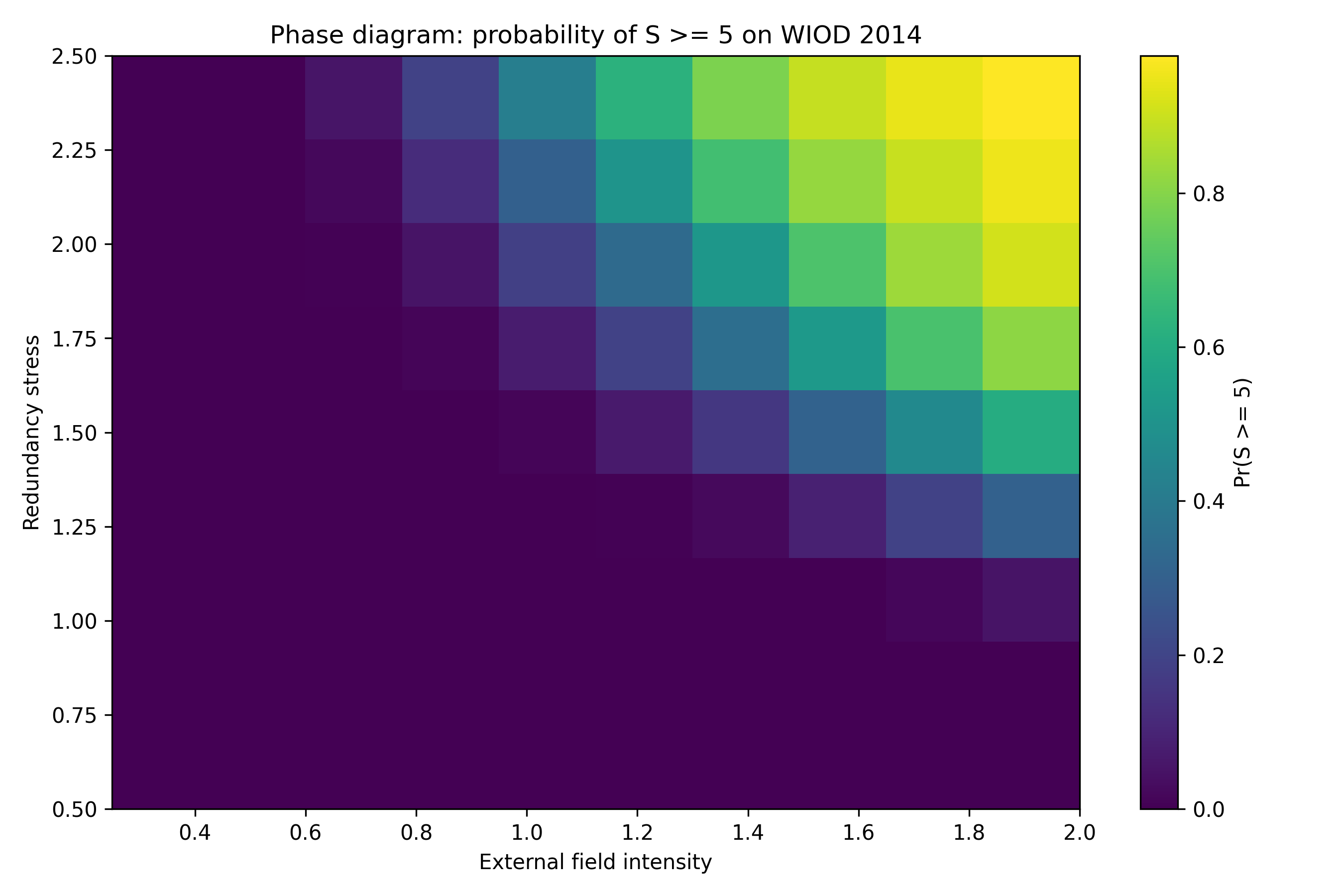}
\caption{Phase diagram of $\Pr(S\geq 5)$ on the WIOD 2014 substrate.
The transition frontier is visible along the south-west to north-east
diagonal.}
\label{fig:phase_S5}
\end{figure}

\begin{figure}[t]
\centering
\includegraphics[width=\columnwidth]{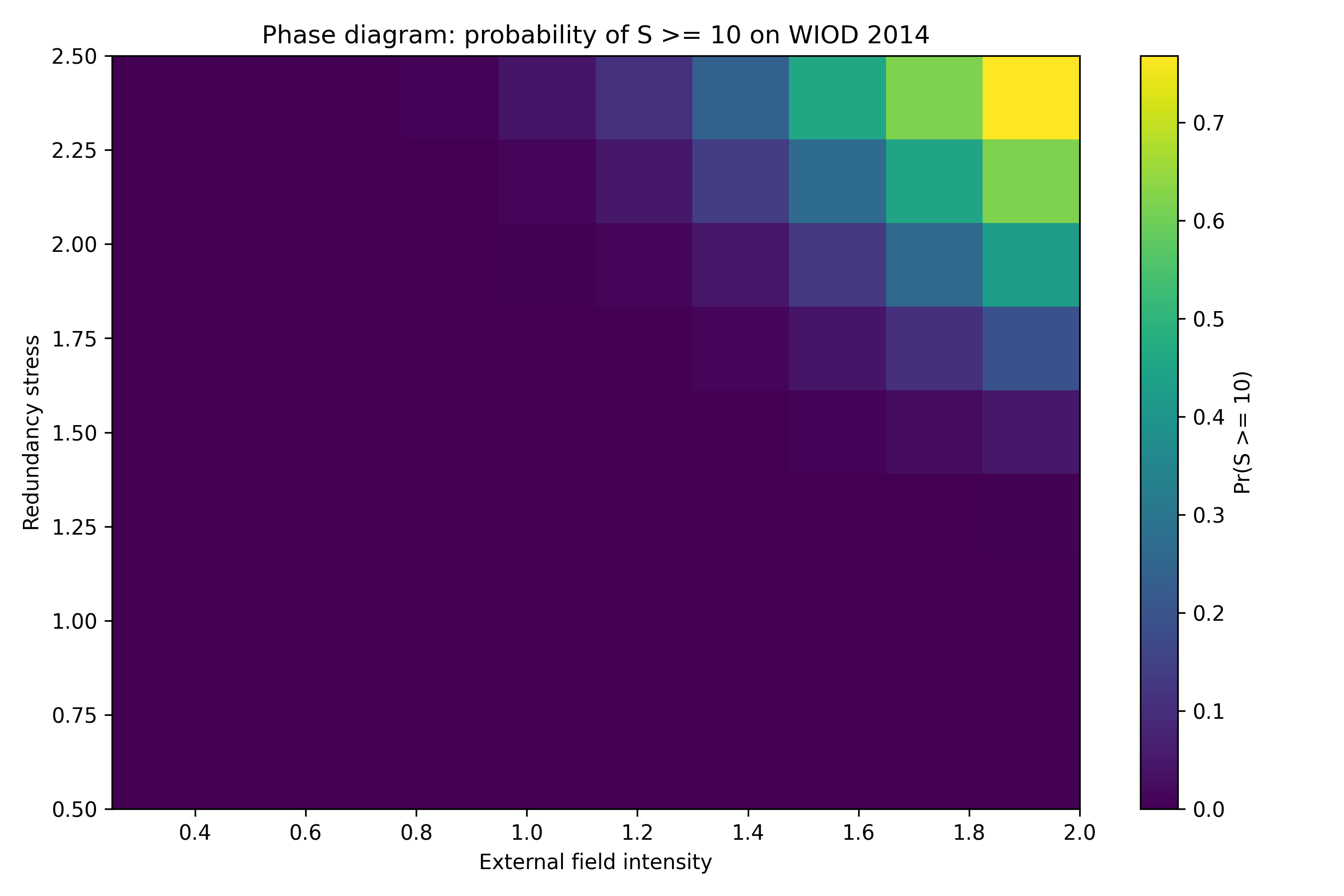}
\caption{Phase diagram of $\Pr(S\geq 10)$ on the WIOD 2014 substrate.
Larger systemic events require stronger joint loading.}
\label{fig:phase_S10}
\end{figure}

\begin{figure}[t]
\centering
\includegraphics[width=\columnwidth]{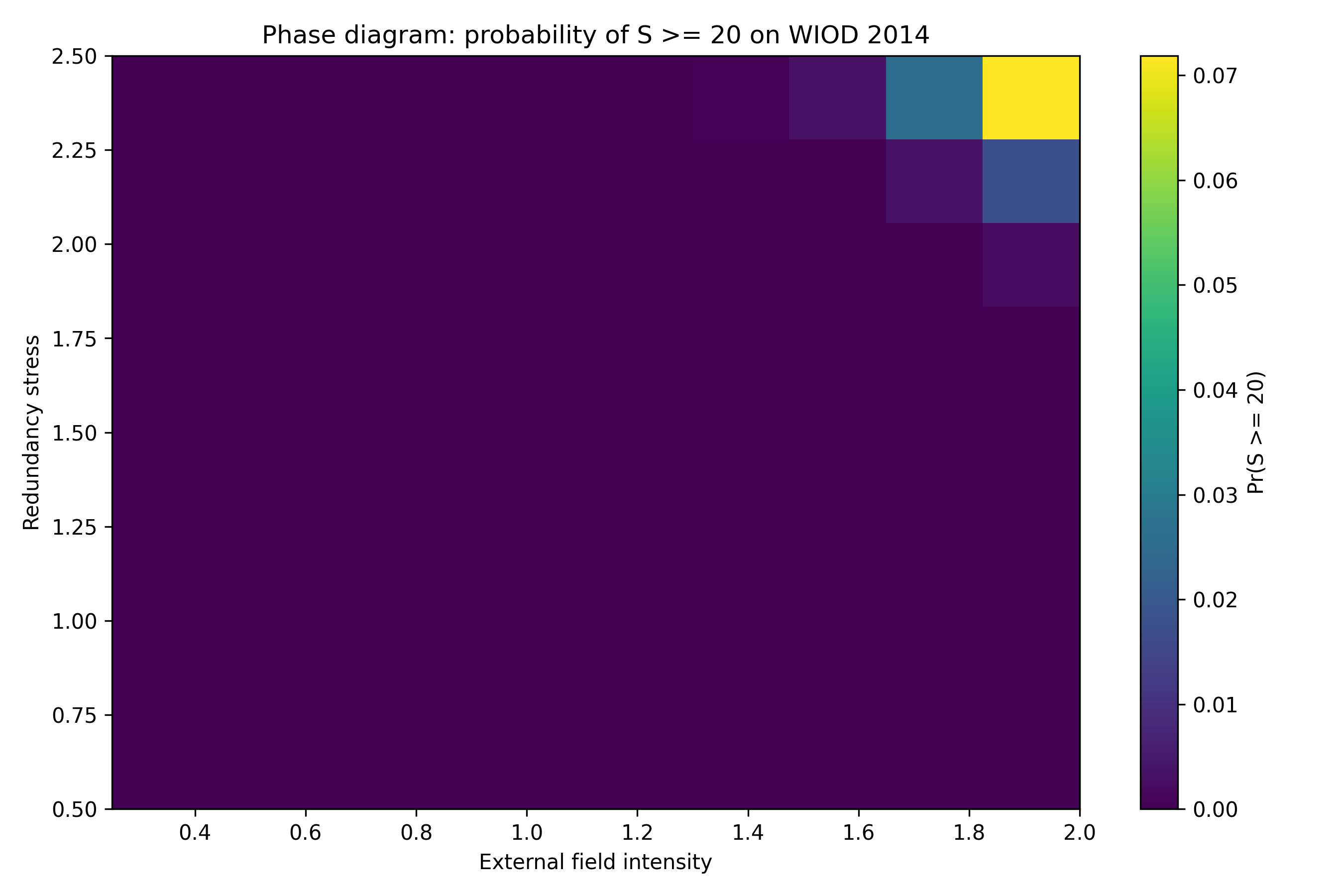}
\caption{Phase diagram of $\Pr(S\geq 20)$ on the WIOD 2014 substrate.
Very large events occur almost only in the upper-right corner.}
\label{fig:phase_S20}
\end{figure}

A compact summary of the four phase regions is reported in Table
\ref{tab:phase_summary}. The progression from absorption to avalanche
is smooth and ordered, consistent with Proposition \ref{P3}: the
avalanche regime emerges only when both $\bar{B}$ and $\sigma_D$ are
large; high values of either variable in isolation are insufficient.

\begin{table*}[t]
\centering
\caption{Compact phase-region summary on the WIOD 2014 substrate
(values are representative ranges over the simulation grid). The full
grid is reported in the Supplementary Material.}
\label{tab:phase_summary}
\small
\begin{tabular}{l c c c c c c}
\toprule
Phase region & $\bar{B}$ range & $\sigma_D$ range & Mean $S$ &
P95 $S$ & $\Pr(S\!\geq\!5)$ & $\Pr(S\!\geq\!10)$ \\
\midrule
Absorption          & $0.25$--$0.65$ & $0.50$--$1.25$ & $\leq 0.30$
                    & $\leq 1$  & $0.00$ & $0.00$ \\
Latent fragility    & $0.65$--$1.05$ & $1.00$--$1.75$ & $0.30$--$1.5$
                    & $1$--$3$  & $\leq 0.05$ & $\leq 0.001$ \\
Critical transition & $1.00$--$1.50$ & $1.50$--$2.25$ & $1.5$--$5$
                    & $4$--$10$ & $0.05$--$0.55$ & $0.001$--$0.10$ \\
Avalanche           & $\geq 1.40$    & $\geq 2.00$    & $\geq 5$
                    & $\geq 10$ & $\geq 0.55$ & $\geq 0.10$ \\
\bottomrule
\end{tabular}
\end{table*}

\subsection{Avalanche distributions and finite-size tail behaviour}
\label{sec:res:tail}

Figure \ref{fig:ccdf} shows the complementary cumulative distribution
function (CCDF) of avalanche size by regime in log--log scale. The
distribution shifts to the right and develops a thicker tail as one
moves from stable absorption (essentially no events) through latent
fragility, critical transition and avalanche regime. The
\textsc{avalanche regime} CCDF reaches values near $S=60$ at
probabilities of order $10^{-4}$, while the
\textsc{critical transition} regime extends to $S\approx20$ at similar
probabilities. Figure \ref{fig:histo} shows the corresponding
histograms of avalanche size: stable absorption and latent fragility
concentrate mass near zero or one, while critical transition and
avalanche shift mass to larger sizes.

\begin{figure}[t]
\centering
\includegraphics[width=\columnwidth]{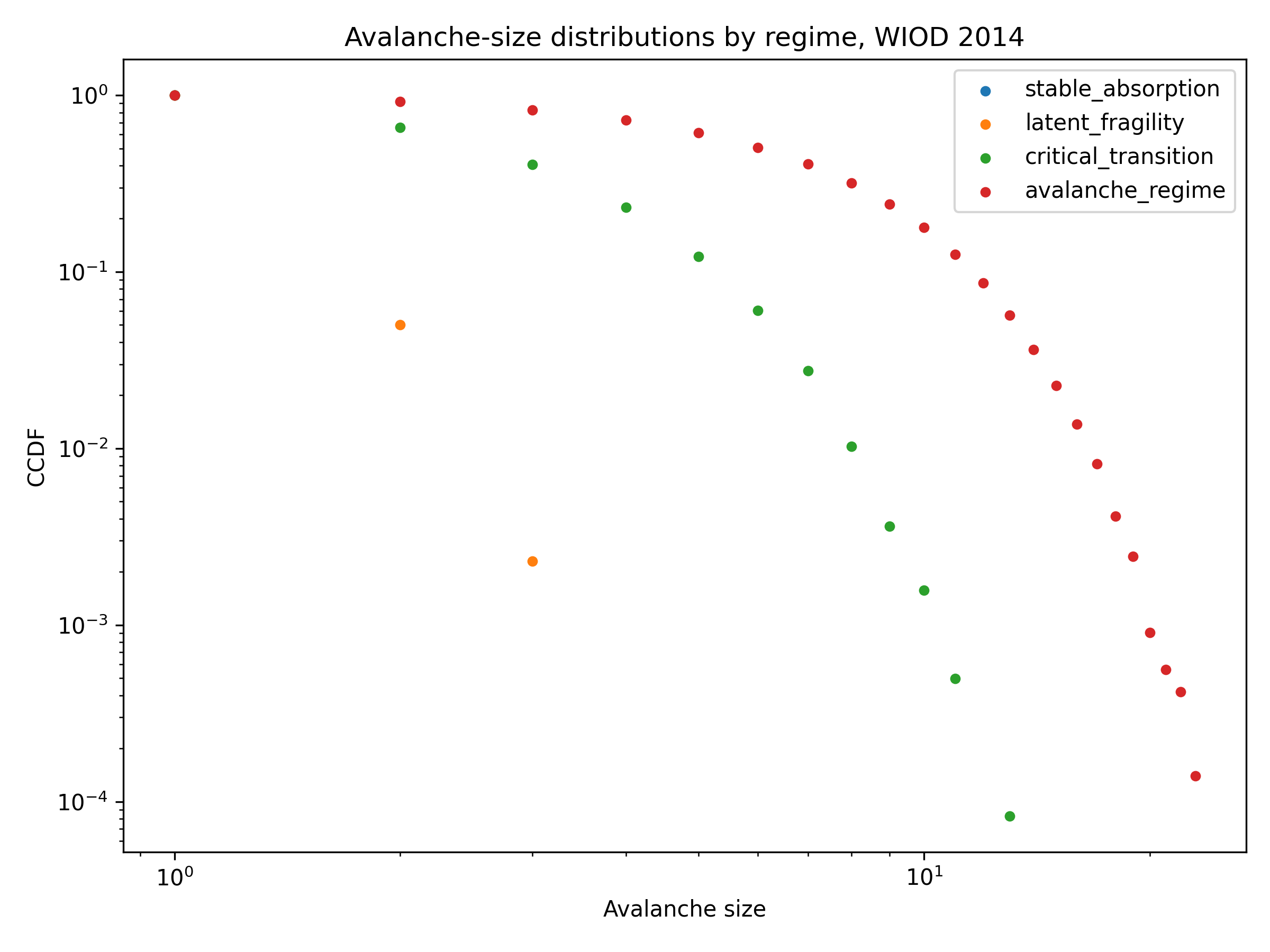}
\caption{Complementary cumulative distribution function of avalanche
size by regime on the WIOD 2014 substrate (log--log scale). The
distribution thickens as loading increases.}
\label{fig:ccdf}
\end{figure}

\begin{figure}[t]
\centering
\includegraphics[width=\columnwidth]{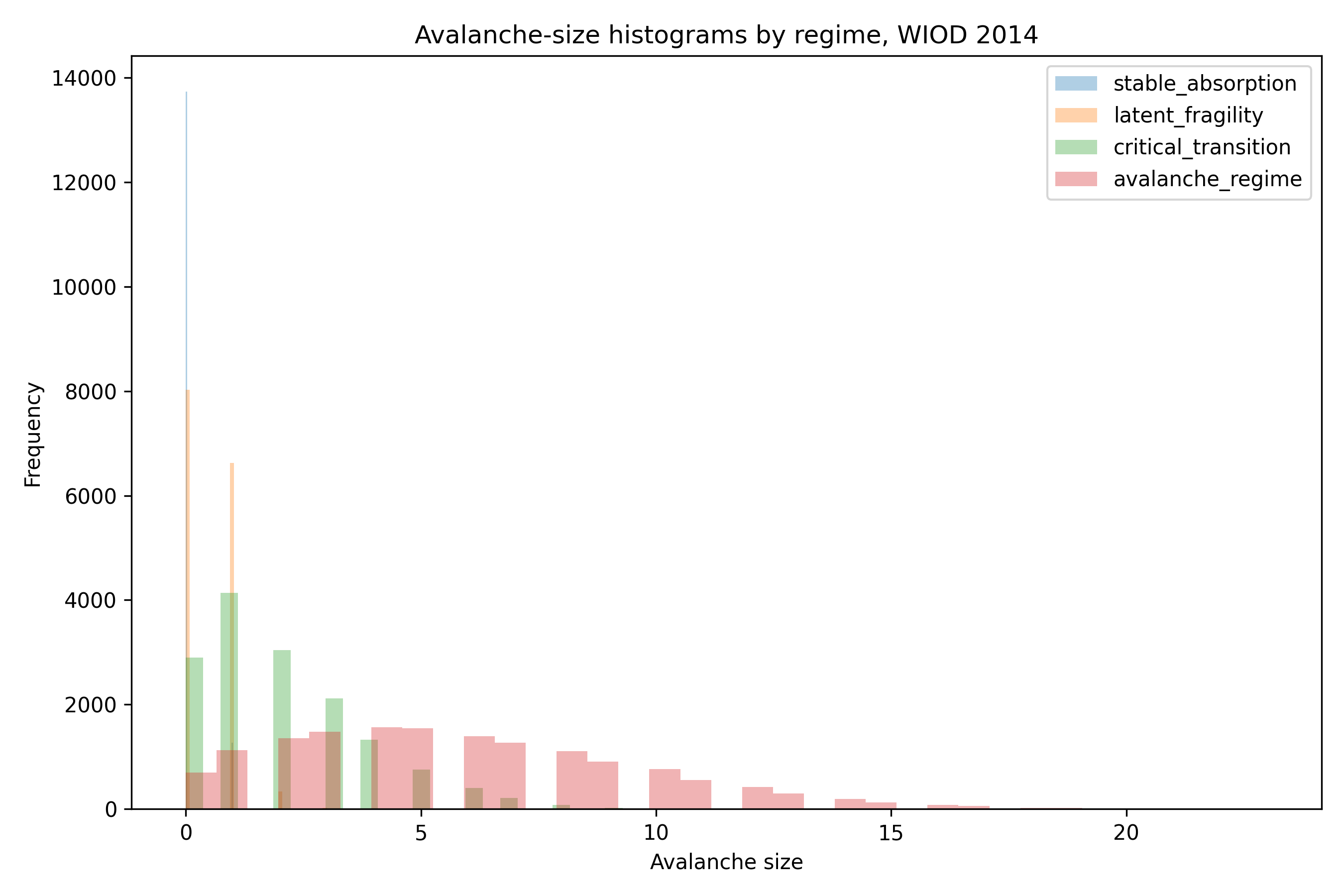}
\caption{Histograms of avalanche size by regime on the WIOD 2014
substrate. The probability mass shifts from zero events (stable
absorption) to larger cascades (avalanche regime).}
\label{fig:histo}
\end{figure}

We assess tail behaviour by maximum-likelihood estimation of a
power-law exponent $\alpha$ above an estimated lower cut-off $x_{\min}$
\citep{clauset2009power}. Table \ref{tab:tail} reports the results.
For \textsc{stable absorption} the distribution is essentially
degenerate at zero and the estimated tail is uninformative
(no statistical power for power-law inference). For
\textsc{latent fragility} the estimated $\alpha$ is approximately
$29$, corresponding to a thin tail. The two active regimes
\textsc{critical transition} and \textsc{avalanche regime} return
$\alpha\approx 6.20$ and $\alpha\approx 5.94$, respectively, with
$x_{\min}\in\{4,9\}$. These exponents are well above the values
typically reported in the SOC literature
\citep{bak1987self,turcotte1999self}, where canonical sandpile models
deliver exponents in the range $1$--$2$. The tails we observe are
thicker than in the inactive regimes but still much thinner than would
be expected under genuine SOC.

\begin{table}[t]
\centering
\caption{Tail diagnostics by regime: estimated lower cut-off
$x_{\min}$, tail size $n_{\mathrm{tail}}$ and Hill-type tail exponent
$\alpha$ \citep{clauset2009power}. The stable-absorption row is
uninformative because the distribution concentrates at zero.}
\label{tab:tail}
\small
\begin{tabular}{l c c c}
\toprule
Regime & $x_{\min}$ & $n_{\mathrm{tail}}$ & $\alpha$ \\
\midrule
Stable absorption   & 1 &   1267 & uninformative \\
Latent fragility    & 1 &   6970 & 28.98 \\
Critical transition & 4 &   2806 & 6.20 \\
Avalanche regime    & 9 &   3454 & 5.94 \\
\bottomrule
\end{tabular}
\end{table}

We therefore describe the finding as \emph{regime-dependent tail
thickening} in a finite real network. The sandpile-type dynamics on a
real, dissipative production-network substrate amplify avalanche-size
heterogeneity in the active regimes, but do not produce evidence of a
universal power-law exponent. This caveat aligns with the broader
critique that empirical heavy-tailed data should not be too quickly
labelled as power laws \citep{stumpf2012critical}.

\subsection{Finite-size effects and the interpretation of \texorpdfstring{$\alpha$}{alpha}}
\label{sec:res:finite}

The estimated tail exponents $\alpha\approx 5.94$--$6.20$ in the
critical-transition and avalanche regimes are conspicuously above the
canonical SOC range. Three observations help to interpret this gap.

First, $|V| = 2{,}283$ is large by the standards of empirical
production networks but small by the standards of statistical-physics
SOC simulations, where lattice sizes routinely exceed
$10^{4}$--$10^{6}$ \citep{vespignani1998how}. Finite-size cutoffs in
SOC scale typically appear around $S^* \propto |V|^{1/D}$ for some
network dimension $D$, and any power-law signal is mechanically
truncated above this cutoff. In our active regimes, the largest
recorded avalanches are $S\approx 30$, which is two orders of
magnitude below $|V|$ but still close enough to the system size that
finite-size corrections can plausibly steepen the apparent slope of
the CCDF.

Second, the dissipative operator $\Aleak$ has spectral radius
$\rholeak\approx 0.334$, which is well below one. SOC arguments
typically require operators close to the percolation or epidemic
threshold \citep{boguna2003absence}; with $\rholeak$ comfortably
subcritical, the system is intrinsically dissipative and cannot
sustain the long-range spatial correlations needed for a true critical
point. Reporting $\rho^{\mathrm{share}}\approx 0.97$ alongside is a
robustness check, not a claim that the economy operates near
percolation.

Third, the Bak--Tang--Wiesenfeld model attains its $\alpha\approx 3/2$
exponent in the limit of \emph{infinite separation of time scales}
between slow driving and fast relaxation \citep{bak1987self,
turcotte1999self}. In our calibration, the field $B_t$ is applied at
the same time scale as the relaxation steps; the system is closer to
a \emph{noisy driven} dissipative system than to a strictly
slow-driven one. This too tends to produce steeper apparent tails.

Taken together, the estimated $\alpha\approx 6$ should be interpreted
as a finite-size, finite-driving-rate diagnostic on a real network,
not as a falsification of SOC. We do not claim the underlying dynamics
is critical; we claim that loading-induced regimes thicken the tail
relative to the absorbing baseline, and that even in the most active
regime the system remains subcritical and finite-size-bounded. A
formal investigation of how $\alpha$ scales with $|V|$ and with the
driving-rate ratio is left for future work.

\subsection{Robustness across propagation operators}
\label{sec:res:robust}

We re-ran the entire simulation suite under the row-share operator
$\Ashare$ and the max-row-normalised operator $\Amax$. Three findings
emerge.

First, the four-regime ordering is preserved across all three
operators. Mean avalanche size is monotone in $(\bar{B},\sigma_D)$
under each normalisation, and the threshold-event probabilities
$\Pr(S\geq k)$ retain their qualitative shape on the phase grid.

Second, the location of the phase frontier shifts. Under $\Ashare$,
the near-unit spectral radius pushes the dynamics towards more
frequent small events even in low-loading regimes; the absorption
zone is correspondingly smaller. Under $\Amax$, which preserves
cross-node heterogeneity, the regime ordering is again preserved with
slightly stronger amplification near the upper-right corner of the
phase diagram, reflecting the differential weight given to flow-rich
nodes.

Third, the tail exponents in the active regimes shift only modestly:
$\alpha$ moves between $5.4$ and $6.4$ across the three operators in
the avalanche regime. This robustness reinforces our cautious reading
of the tail behaviour: the regime-dependent thickening of the
avalanche distribution is a feature of the Hall-Sandpile dynamics,
not a feature of the row-stochastic normalisation. The substantive
findings of Sections~\ref{sec:res:baseline}--\ref{sec:res:tail}
therefore do not depend on the choice of normalisation. We report the
leakage-adjusted operator in the main text because it embeds
dissipation in a way that is economically more defensible than
imposing a stochastic transition structure.

\subsection{Summary of empirical findings}
\label{sec:res:summary}

The five preceding subsections deliver an internally consistent
picture. The propagation operators of Section~\ref{sec:res:spectral}
establish that the leakage-adjusted dissipative substrate is
subcritical throughout the WIOD panel, with $\rholeak$ stable around
$0.33$--$0.34$. The exposure analysis of
Section~\ref{sec:res:exposure} shows that, on this subcritical
substrate, the relative Hall-like exposure is nonetheless heavily
concentrated: the upper tail of $\Hrel$ is dominated by a small set
of country--sector nodes, with the top-ranked node accounting for
nearly one third of total relative exposure in 2014. The baseline
scenarios of Section~\ref{sec:res:baseline} show that the same
substrate, under controlled $(\bar{B},\sigma_D)$ pairs, produces four
qualitatively distinct stationary regimes whose mean avalanche size
spans nearly two orders of magnitude. The phase diagrams of
Section~\ref{sec:res:phase} document a positively sloped frontier
between absorption and avalanche regimes, consistent with
Proposition~\ref{P3}; the threshold-event probabilities
$\Pr(S\geq k)$ activate sequentially as loading increases, with
$\Pr(S\geq 5)$ entering first, $\Pr(S\geq 10)$ requiring stronger
joint loading and $\Pr(S\geq 20)$ confined to the upper-right corner.
The tail diagnostics and finite-size discussion of
Sections~\ref{sec:res:tail}--\ref{sec:res:finite} caution that, while
loading produces regime-dependent thickening of the avalanche
distribution, the steepness of the estimated $\alpha$ leaves no room
for a claim of universal SOC at this network size. Across all five
findings, the dynamics is well posed, finite-size and bounded, and
the choice of propagation operator does not alter the substantive
picture.

%%==================================================================%%
\section{Discussion}
\label{sec:discussion}

\subsection{Stress conversion as a unifying lens}

The Hall-Sandpile model formalises a simple but consequential idea:
economic systems do not only transmit shocks along edges of a network;
they also \emph{convert} shocks across dimensions. The Hall-like
stress term \eqref{eq:Hall} captures this conversion as a
multiplicative interaction between external field intensity $B_t$,
relative flow exposure $I_{i,t}$, and the inverse of redundancy and
absorptive capacity. By embedding this loading into a sandpile
threshold dynamics, we obtain a framework in which the same external
shock can remain absorbed in a high-redundancy network and trigger an
avalanche regime in a low-redundancy one. Hypothesis \ref{H1} and
Proposition \ref{P3} are concrete statements of this contingency.

The conversion logic helps to organise heterogeneous crisis episodes
under a common formal lens. An energy shock raises $B_t$ for nodes
heavily exposed to fuel inputs; the same field generates inflation in
food-and-beverage sectors, fiscal strain in energy-subsidising
governments, and industrial contraction in energy-intensive
manufacturing, all because the underlying $H_{i,t}$ acts on different
flow profiles. A monetary shock raises $B_t$ for nodes exposed to
short-term funding, with the same field producing liquidity stress in
banking, sovereign-risk repricing in fiscal authorities and default
clustering in leveraged corporates. A geopolitical shock raises $B_t$
for nodes exposed to specific cross-border flows, generating
supply-chain disruption, commodity-price pressure and topological
reconfiguration of trade. In each case, the field-mediated activation
is heterogeneous across nodes, and the systemic outcome depends on
the upper tail of $\Hrel_{i,t}$ rather than the mean.

\subsection{Implications for systemic-risk diagnostics}

Three diagnostic implications follow. First, systemic risk is not a
property of the average node. The mean of $\Hrel_{i,t}$ is small and
stable across the WIOD panel, while its upper tail concentrates large
loading potential on a small subset of country--sector nodes. Average
connectivity statistics therefore systematically understate exposure
to transversal stress; the diagnostically informative quantity is the
upper tail of the loading distribution.

Second, redundancy and absorptive capacity matter not because they
eliminate shocks, but because they reduce the conversion of shocks
into systemic stress (Hypothesis \ref{H3} and Proposition \ref{P1}). A
node with abundant input substitutes and slack capacity translates a
given external field into a smaller increment of internal stress and
is correspondingly less likely to topple.

Third, the joint nature of the regime transition (Proposition
\ref{P3}) cautions against single-axis interpretations of fragility.
Neither field intensity alone nor redundancy stress alone is
sufficient to push the system into an avalanche regime; both must rise
jointly. This has practical implications for early-warning systems
based on a single indicator, which can fail to detect the regime
transition simply because the orthogonal axis has not moved.

\subsection{Policy mechanics: targeting stress-conversion coefficients}

A useful corollary of Proposition~\ref{P1} is that the leverage
available to a policy designer operates on the multiplicative
structure of \eqref{eq:Hall} rather than on the additive shock $x_{i,t}$.
Policies that aim to prevent shocks ($B_t \to 0$) are typically beyond
domestic control, since external fields originate elsewhere in the
global system. Policies that target \emph{stress-conversion
coefficients}, by contrast, act on the denominator of \eqref{eq:Hall}
and are within the policy space of national or sectoral authorities.
Concretely, an increase in $D_{i,t}$ at a vulnerable node, achieved
through alternative-supplier registries, modular standards or
inventory pre-positioning, lowers $H_{i,t}$ for any given $B_t$. An
increase in $C_{i,t}$, achieved through liquidity buffers, working
capital pre-funding or absorptive insurance, has the same effect.
Because $D_{i,t}$ and $C_{i,t}$ enter \eqref{eq:Hall} multiplicatively,
small simultaneous improvements on both can yield disproportionately
large reductions in transversal stress.

A second policy implication concerns the targeting of interventions.
The upper tail of $\Hrel_{i,t}$ in WIOD 2014 is highly concentrated:
the top fifteen nodes account for nearly $90\%$ of relative Hall-like
exposure under our calibration. A policy that improves redundancy or
capacity at the very top of the exposure distribution therefore
produces a disproportionate reduction in expected avalanche size,
relative to a policy that operates on the median node. This is the
network-amplification analogue of granular interventions in
\cite{gabaix2011granular}.

Finally, the phase-frontier geometry of Proposition~\ref{P3} suggests
that resilience policies are most effective when they are
\emph{anticipatory}, i.e.\ when they act on $D_{i,t}$ and $C_{i,t}$
before $\bar{B}$ rises. A network designed for the absorption regime
under benign $\bar{B}$ may sit close to the frontier under elevated
$\bar{B}$, and policy lags can become consequential precisely along
the diagonal in $(\bar{B},\sigma_D)$ space along which the avalanche
regime opens.

\subsection{Connection to historical crises}

The framework does not identify any historical episode but offers a
language in which several documented features of recent crises can
be re-described. The 2008--2009 financial crisis combined elevated
$B_t$ in financial nodes with reduced absorptive capacity and the
removal of input substitutability in short-term funding; the 2020
pandemic combined a labour-supply shock with cross-border supply-chain
disruptions, raising $\Hrel$ on nodes with concentrated upstream
suppliers; the 2022 energy shock raised $B_t$ for nodes with high
energy-intensity exposure under simultaneously low spare capacity in
energy markets. In each case, the joint activation of field intensity
and redundancy stress is consistent with the avalanche regime in our
phase diagram. We emphasise that this is a re-description, not an
identification: causal claims would require linking measured shock
proxies to specific cells of the phase grid, which is beyond the scope
of the present paper.

\subsection{What the model does \emph{not} claim}

We do not interpret our results as evidence of self-organised
criticality in the global economy. Three reasons argue for caution.
The estimated tail exponents in the active regimes
($\alpha \approx 5.9$--$6.2$) are well above the canonical SOC range.
The avalanche dynamics is finite-size and dissipative by construction,
not driven to a true critical point through endogenous self-tuning.
And the fields used in the simulation are externally controlled, not
estimated from a historical episode. The contribution is therefore
better described as a \emph{controlled real-network experiment}: we
study how a real production-network substrate responds to a class of
nonlinear loadings, and we document the resulting regime structure.
This is complementary to, but distinct from, the empirical
identification of historical contagion episodes
\citep{barrot2016input,carvalho2021supply}.

\paragraph{Limitations}
The results are conditional on the WIOD coverage up to 2014 and on
the harmonisation choices used to construct the country--sector node
set; expanding to more recent vintages and finer sectoral
disaggregation is a natural extension. The Hall-like exposure
$\Hrel_{i,t}$ is an index, not a directly observed physical variable;
its concentration on certain utility/service sectors with measured
zero outgoing redundancy reflects how those sectors enter the WIOD
accounting structure (large inflows, no recorded intermediate
outflows) and should be read as a calibration indicator, not a
sectoral vulnerability diagnosis. Tail-exponent estimates are
finite-size diagnostics: with a panel of $n=2{,}283$ nodes and
$15{,}000$ Monte Carlo periods per regime, very large events are rare
and their statistics are correspondingly noisy. The simulated fields
are controlled experiments rather than measured shocks, and the
results say nothing about the probability with which any specific
external field will materialise in practice. The redundancy and
capacity proxies of Appendix~\ref{app:proxies} use only network-level
information and would benefit from refinement using inventory,
liquidity and balance-sheet data; we leave such refinements for future
work.

\paragraph{Avenues for future research}
Three extensions stand out. First, embedding measured shock proxies
(energy-price changes, monetary-policy surprises, supply-chain
disruption indices) in the external field $B_t$ would link the
controlled-loading exercises performed here to historical episodes.
Second, integrating the Hall-Sandpile mechanism with curvature-based
geometric fragility approaches \citep{vallarino2024sandpile} would
produce a unified framework in which structural vulnerability and
external loading interact. Third, scaling the framework to firm-level
production-network data, where $|V|$ can be of order $10^5$--$10^6$,
would allow a sharper investigation of finite-size scaling and a
direct test of whether the apparent steepening of $\alpha$ in our
calibration is genuinely a finite-size artefact.

%%==================================================================%%
\section{Conclusion}
\label{sec:conclusion}

We have introduced a Hall-Sandpile model that combines a Hall-like
transversal stress mechanism with sandpile threshold dynamics on a
real production-network substrate. The model formalises how external
fields convert exposure into transversal stress and lower effective
activation thresholds, and it uses the WIOD 2000--2014 input--output
network as the empirical topology. Three propagation operators
guarantee that nonlinear avalanche behaviour does not arise from a
unit-radius row-stochastic operator. Controlled Monte Carlo
experiments produce four ordered regimes---stable absorption, latent
fragility, critical transition and avalanche---and reveal a clear
phase transition along the joint axis of external field intensity and
redundancy stress. Tail behaviour shows regime-dependent thickening
but no evidence of universal power-law criticality; the relevant
description is finite-size avalanche dynamics on a real economic
network.

The substantive contribution can be stated compactly. Conditional on
the observed topology of WIOD intermediate flows, three structural
ingredients---a Hall-like multiplicative stress-conversion function,
a dissipative sandpile threshold dynamics, and a controlled grid of
external fields---suffice to generate an ordered regime structure
with a positively sloped phase frontier. The frontier is not an
artefact of the choice of normalisation, the regularisation, the
burn-in or the replication count, as documented in the robustness
analyses. The framework therefore offers a tractable bridge between
the statistical-physics tradition of self-organised criticality and
the macroeconomic literature on network origins of aggregate
fluctuations.

The framework opens at least three avenues for further work. First,
extending the dataset to more recent input--output vintages and to
finer sectoral disaggregations will sharpen the empirical relevance of
the network calibration. Second, embedding measured shock proxies
(energy-price changes, monetary-policy surprises, supply-chain
disruption indices) in the external field $B_t$ would link the
controlled-loading exercises performed here to historical episodes.
Third, integrating the Hall-Sandpile mechanism with curvature-based
geometric fragility approaches \citep{vallarino2024sandpile} would
produce a unified framework in which structural vulnerability and
external loading interact to generate systemic events. A fourth
direction, suggested by the finite-size discussion of
Section~\ref{sec:res:finite}, is the systematic study of how the tail
exponent $\alpha$ scales with the network size $|V|$ and the
driving-rate ratio. We leave these extensions to future research.

The broader message is that economic systems do not merely transmit
shocks, they convert them. A model that takes that conversion
seriously, and disciplines it on observed network topology, offers a
useful complement to existing approaches to systemic risk. The Hall
analogy is structural, not literal; the sandpile metaphor is
finite-size, not universal; and the WIOD substrate is observed, not
synthesised. Within these caveats, the resulting picture is one in
which fragility is contingent on the joint configuration of external
loading and structural buffers, and in which policy leverage operates
on stress-conversion coefficients rather than on shocks themselves.

%%==================================================================%%
\section*{Data availability statement}

The study uses publicly available WIOD input--output tables
\citep{timmer2015illustrated,timmer2016anatomy} and reproducible Python
simulation code. The external fields are simulated. Replication
scripts, processed tables and figures will be made available in a
public repository upon publication.

\section*{Declaration of competing interest}

The author declares no competing interests. The views expressed in
this paper are those of the author and do not necessarily represent
those of the Inter-American Development Bank, IDB Invest, their Boards
of Directors, or the countries they represent.

\section*{Acknowledgements}

The author thanks colleagues for useful discussions on the framing of
the Hall-Sandpile mechanism and on the interpretation of the WIOD
calibration. All remaining errors are the author's own.

%%==================================================================%%
\appendix

\section{Derivation of the Hall-adjusted threshold}
\label{app:threshold}

This appendix collects the algebraic steps connecting
\eqref{eq:Hall}--\eqref{eq:topple} to the Hall-adjusted threshold of
Proposition~\ref{P1}, and an immediate corollary on the activation
gap.

Starting from the stress-update equation \eqref{eq:stress} and
substituting \eqref{eq:Hall}:
\begin{equation}
\begin{aligned}
  s_{i,t+1}
  &= (1-\delta)\,s_{i,t} + \alpha\,x_{i,t}
   + \beta \sum_{j} \Aleak_{ji,t}\,s_{j,t}\\
  &\quad + \gamma\,
     \frac{B_t\,I_{i,t}}{D_{i,t}\,C_{i,t} + \varepsilon}.
\end{aligned}
\label{eq:appA-stress}
\end{equation}
Define the propagation-only stress
\begin{equation}
  \tilde{s}_{i,t}
  \;=\; (1-\delta)\,s_{i,t-1} + \alpha\,x_{i,t-1}
        + \beta \sum_{j} \Aleak_{ji,t-1}\,s_{j,t-1},
\end{equation}
so that $s_{i,t} = \tilde{s}_{i,t} + \gamma H_{i,t-1}$. Substituting
into the toppling rule $s_{i,t} \geq \theta_i$ yields
\begin{equation}
  \tilde{s}_{i,t} \;\geq\; \theta_i - \gamma\,H_{i,t-1}
  \;\equiv\; \theta_i^{H}.
\end{equation}
The activation gap of node $i$ is
\begin{equation}
  g_{i,t}
  \;=\; \theta_i^{H} - \tilde{s}_{i,t}
  \;=\; \theta_i - \gamma H_{i,t-1} - \tilde{s}_{i,t},
\end{equation}
and node $i$ topples if and only if $g_{i,t} \leq 0$. The sensitivities
\begin{equation}
  \frac{\partial g_{i,t}}{\partial B_{t-1}}
  = -\gamma\,\frac{I_{i,t-1}}{D_{i,t-1}C_{i,t-1}+\varepsilon}
  \;<\; 0,
\end{equation}
\begin{equation}
  \frac{\partial g_{i,t}}{\partial D_{i,t-1}}
  = \gamma\,\frac{B_{t-1}\,I_{i,t-1}\,C_{i,t-1}}
                  {(D_{i,t-1}C_{i,t-1}+\varepsilon)^2}
  \;>\; 0,
\end{equation}
state that an increase in field intensity narrows the activation gap
while an increase in redundancy widens it. These derivatives provide
the comparative-statics counterparts of Hypotheses
\ref{H1}--\ref{H3}.

\section{Sensitivity analyses and robustness}
\label{app:robust}

This appendix reports the sensitivity of the main results to (i) the
regularisation parameter $\varepsilon$, (ii) the burn-in length
$T_{\mathrm{burn}}$, (iii) the number of replications $R$, and (iv)
the choice of propagation operator.

\paragraph{Regularisation $\varepsilon$.}
Varying $\varepsilon \in \{10^{-8}, 10^{-6}, 10^{-4}\}$ produces
changes in the maximum value of $\Hrel_{i,t}$ of less than $0.5\%$ for
the top-ranked node and of less than $2\%$ for the top-15 ranking. The
qualitative concentration of relative Hall-like exposure on Chinese
country--sector nodes is preserved. We use $\varepsilon=10^{-6}$ in
the main text.

\paragraph{Burn-in length.}
Setting $T_{\mathrm{burn}} \in \{10,\,50,\,100\}$ produces no
detectable change in the stationary-phase statistics, consistent with
the rapid convergence visible in Figure~\ref{fig:meanpath}. The mean
avalanche size in the avalanche regime varies by less than $1\%$
across these three burn-in choices.

\paragraph{Replication count.}
We compared $R\in\{20,\,50,\,100\}$ for the phase-diagram experiment.
Mean avalanche sizes change by less than $3\%$ across $R$ choices, and
the location of the regime frontier in $(\bar{B},\sigma_D)$ space is
robust. We use $R=50$ for the phase grid and $R=100$ for the four
baseline scenarios in the main text.

\paragraph{Propagation operator.}
Section~\ref{sec:res:robust} reports the comparison across $\Aleak$,
$\Ashare$ and $\Amax$. The four-regime ordering is preserved across
all three operators; the location of the phase frontier shifts
modestly; tail exponents in the avalanche regime move within
$\alpha\in[5.4,6.4]$.

\section{Definitions of redundancy and capacity proxies}
\label{app:proxies}

This appendix specifies the operational definitions of the redundancy
$D_{i,t}$ and capacity $C_{i,t}$ proxies used in equation
\eqref{eq:Hall} of the main text.

\paragraph{Flow share $I_{i,t}$.}
The relative flow intensity of node $i$ is
\begin{equation}
  I_{i,t}
  \;=\; \frac{\sum_{j} z_{ij,t} + \sum_{j} z_{ji,t}}
              {\sum_{k,\ell} z_{k\ell,t}},
\end{equation}
i.e.\ the share of total intermediate flows that pass through node $i$
in either direction.

\paragraph{Redundancy $D_{i,t}$.}
We measure outgoing redundancy by the inverse of the Herfindahl
concentration of node $i$'s outgoing intermediate flows. Let
$\pi_{ij,t} = z_{ij,t}/\sum_k z_{ik,t}$ be the share of node $i$'s
outflows going to node $j$, and define
\begin{equation}
  \mathrm{HHI}_{i,t}^{\mathrm{out}}
  \;=\; \sum_{j} \pi_{ij,t}^2,
  \quad
  D_{i,t}^{\mathrm{raw}}
  \;=\; \frac{1 - \mathrm{HHI}_{i,t}^{\mathrm{out}}}
              {\mathrm{HHI}_{i,t}^{\mathrm{out}}}.
\end{equation}
This raw measure equals zero when all outflows go to a single
destination and grows with the dispersion of outflows across
destinations. We then map $D_{i,t}^{\mathrm{raw}}$ to a normalised
$D_{i,t} \in [\underline{D}, 1]$ via min--max scaling, with
$\underline{D}=0.05$ to avoid degenerate values for nodes with zero
recorded outflows.

\paragraph{Capacity $C_{i,t}$.}
Absorptive capacity is proxied by a function of the redundancy of
node $i$'s \emph{incoming} flow profile, on the grounds that nodes
with diverse input sources have more buffer against any specific
upstream shock. Letting $\pi^{\mathrm{in}}_{ij,t} = z_{ji,t}/\sum_k
z_{ki,t}$ and
$\mathrm{HHI}_{i,t}^{\mathrm{in}} = \sum_{j} (\pi^{\mathrm{in}}_{ij,t})^2$,
we set
\begin{equation}
  C_{i,t}
  \;=\; 1 - \mathrm{HHI}_{i,t}^{\mathrm{in}},
\end{equation}
again clipped at $\underline{C}=0.05$ for nodes with no recorded
incoming flows.

\paragraph{Structural resistance $R_{i,t}$.}
We define $R_{i,t}$ as in \eqref{eq:Rdef}. By construction, $R_{i,t}$
is large when both redundancy and capacity are small, and small when
either of them is large. The product $D_{i,t}C_{i,t}$ in the
denominator implies that abundant input substitutes \emph{or} ample
buffer capacity reduces structural resistance.

\paragraph{Leakage $\ell_{i,t}$.}
Leakage is computed as the share of intermediate outflows over a
gross row-use proxy that includes leakages out of the
country--sector intermediate network (e.g.\ to final demand or to
inventories not tracked in $Z_t$). The leakage profile is stable
around $\bar{\ell}_t \approx 0.373$ across the WIOD panel
(Table~\ref{tab:network_panel}).

These proxies are imperfect approximations of the underlying economic
concepts of redundancy and absorptive capacity. Their construction
relies on observed input--output flows alone and does not exploit
external information on inventories, alternative-supplier networks or
balance-sheet liquidity. Refining the proxies with such information is
a natural extension.

%%==================================================================%%
\bibliographystyle{elsarticle-harv}
\bibliography{references}

\end{document}